\documentclass[lettersize,journal]{IEEEtran}
\usepackage{amsmath,amsfonts}
\usepackage{cite}
\usepackage{array}
\usepackage{textcomp}
\usepackage{amsthm}
\usepackage[english]{babel}
\newtheorem{theorem}{Theorem}
\newtheorem{lemma}{Lemma}
\newtheorem{remark}{Remark}
\newtheorem{definition}{Definition}
\newtheorem{assumption}{Assumption}
\usepackage{stfloats}
\usepackage{graphicx}
\usepackage[caption=false,font=footnotesize]{subfig}
\usepackage{caption}
\DeclareCaptionLabelSeparator{period}{. }
\usepackage{url}
\usepackage{verbatim}
\usepackage{graphicx}
\usepackage{amssymb}
\usepackage{algorithm2e}
\usepackage{amsmath}
\allowdisplaybreaks
\RestyleAlgo{ruled}
\def\BibTeX{{\rm B\kern-.05em{\sc i\kern-.025em b}\kern-.08em
    T\kern-.1667em\lower.7ex\hbox{E}\kern-.125emX}}
\usepackage{balance}
\usepackage{xcolor}

\usepackage{xcolor} 
\usepackage{soul}   

\definecolor{lightblue}{RGB}{173,216,230} 
\sethlcolor{lightblue}                    

\begin{document}
\title{Adaptive Coded Federated Learning: Privacy Preservation and Straggler Mitigation}
\author{Chengxi Li, Ming Xiao,~\IEEEmembership{Senior Member,~IEEE,} and Mikael Skoglund,~\IEEEmembership{Fellow, IEEE}
\thanks{C. Li, M. Xiao and M. Skoglund are with the Division of Information Science and Engineering, School of Electrical Engineering and Computer Science, KTH Royal Institute of Technology, 10044 Stockholm, Sweden. (e-mail: chengxli@kth.se; mingx@kth.se; skoglund@kth.se). Corresponding author: Chengxi Li.}
}

\markboth{Journal of \LaTeX\ Class Files}%
{}

\maketitle

\begin{abstract}
In this article, we address the problem of federated learning in the presence of stragglers. For this problem, a coded federated learning framework has been proposed, where the central server aggregates gradients received from the non-stragglers and gradient computed from a privacy-preservation global coded dataset to mitigate the negative impact of the stragglers. However, when aggregating these gradients, fixed weights are consistently applied across iterations, neglecting the generation of the global coded dataset and the dynamic nature of the trained model over iterations. This oversight may result in diminished learning performance. To overcome this drawback, we propose a new method named adaptive coded federated learning (ACFL). In ACFL, before the training, each device uploads a local coded dataset with additive noise to the central server to generate a global coded dataset under privacy-preservation requirements. During each iteration of the training, the central server aggregates the gradients received from the non-stragglers and the gradient computed from the global coded dataset, where an adaptive policy for varying the aggregation weights is designed. Under this policy, we optimize the performance in terms of privacy and learning, where the learning performance is analyzed through convergence analysis and the privacy performance in sharing local coded datasets with the server is characterized via mutual information differential privacy. Finally, we perform simulations to demonstrate the superiority of ACFL compared with the baseline methods. 
\end{abstract}


\section{Introduction}
\label{introduction} 
\IEEEPARstart{I}{n} intelligent systems, edge devices are generating an increasing amount of data, which can be utilized to train machine learning models for various tasks across a wide range of applications, such as smart cities \cite{band2022smart}, autonomous driving \cite{balkus2022survey}, and intelligent healthcare \cite{abdulkareem2021realizing}. Traditionally, these data are uploaded to a central server for centralized machine learning. However, uploading datasets directly to the central server without privacy-preservation measures raises severe privacy concerns for the edge devices \cite{zhang2023multi,du2020machine}. As an alternative to centralized machine learning, federated learning (FL) has emerged as a highly effective tool \cite{du2024distributed,khan2021federated, ye2023heterogeneous, shao2023privacy}. Under the FL framework, during each iteration of the training process, local devices compute gradients with their local datasets and send these gradients to a central server. The central server then aggregates the received gradients to update the global model \cite{mcmahan2017communication}. In this manner, the edge devices only share locally computed gradients with the central server, significantly enhancing the privacy of raw datasets. 

In FL, the training process may be hindered by devices known as stragglers, which are significantly slower than others or completely unresponsive due to various incidents \cite{karakus2017straggler, egger2023sparse,chengxi2024preprint}. Additionally, in FL, the local datasets on different devices can be highly heterogeneous, meaning they are not independent and identically distributed (non-i.i.d.) across devices \cite{vahidian2023rethinking, abdelmoniem2023comprehensive,li2021federated,li2021decentralized}. In this case, the presence of stragglers can bias the trained global model due to the absence of the contributions of some devices in each iteration. In other words, stragglers can negatively impact the learning performance of FL, potentially leading to poor convergence. In traditional distributed learning scenarios, the challenges posed by stragglers can be addressed through well-known gradient coding (GC) strategies \cite{tandon2017gradient,Chengxi20231-bitGC, wang2019erasurehead, ozfatura2019gradient, buyukates2022gradient, glasgow2021approximate,bitar2020stochastic}. The core concept of GC is replicating the training data samples and distributing them redundantly across the devices prior to training. During each training iteration, non-straggler devices encode the locally computed gradients and transmit the encoded vectors to the central server.  Upon receiving these vectors, the central server can reconstruct the exact global gradient by leveraging the redundancy in the training data distribution. In this way, the impact of stragglers is mitigated, as the missing information from stragglers can be compensated for using the encoded vectors transmitted by non-stragglers.  However, in FL scenarios, privacy concerns make it impractical to adopt GC techniques and to allocate training data to devices redundantly in a carefully planned manner. This is because devices are reluctant to share their local datasets directly with others to create data redundancy considering that their local datasets may contain sensitive or private information \cite{li2021survey}. 
 Another line of research focused on mitigating stragglers in FL is the development of asynchronous FL algorithms, where the server updates the global model without waiting for all devices to report in every iteration. Instead, the server incorporates stale gradients received from slower devices, thereby decoupling the update process from synchronous coordination. For instance, in \cite{nguyen2022federated}, a buffered asynchronous aggregation method is proposed, where the server maintains a buffer and performs a model update only when a predefined number of messages from local devices have been collected. In \cite{hu2023scheduling}, a periodic aggregation protocol is introduced that jointly considers training data representation and channel qualities. 
However, the performance gain of asynchronous methods is inherently limited, since the information from stragglers is less likely to be frequently utilized by the server, and the missing or stale information from stragglers negatively impacts overall learning performance.

For FL problems, to evade the negative impact of stragglers while avoiding the drawbacks of the aforementioned schemes, a new paradigm named coded federated learning (CFL) has been investigated \cite{dhakal2019coded, prakash2020coded, anand2021differentially, sun2023stochastic}. In CFL, each device uploads a coded version of its dataset to the central server to form a global coded dataset before training begins. During the training process, the central server aggregates gradients received from the non-stragglers and gradient computed from the global coded dataset to mitigate the negative impact of the stragglers.  Compared with sharing raw datasets as in GC, by carefully designing the way of encoding local datasets in CFL, privacy concerns can be largely reduced, and one can even achieve perfect privacy. In addition, compared with asynchronous FL methods, the learning performance can be enhanced in CFL based on the fact that the missing information from the stragglers during the training iterations can be compensated by the information gained from the global coded dataset at the server.  The key issue in CFL is how to generate the global coded dataset before training, as well as how to aggregate the gradients at the central server during the training process, in order to achieve better performance in terms of learning and privacy.  

There has been some recent progress under the framework of CFL. To be more specific, CFL was firstly investigated in \cite{dhakal2019coded} for linear regression problems under the FL setting, where local coded datasets are constructed at the devices by applying random linear projections. In \cite{prakash2020coded}, the CFL scheme was combined with distributed kernel embedding to make CFL applicable for nonlinear FL problems. Later, to further enhance the privacy performance of CFL in sharing local coded datasets with the server, a differentially private CFL method was proposed by adding noise to the local coded datasets, where the differential privacy guarantee was proved \cite{anand2021differentially}. Very recently, a stochastic coded federated learning (SCFL) method has been proposed \cite{sun2023stochastic}. In SCFL, before the training starts, the devices generate and upload privacy-preservation coded datasets to the central server by adding noise to the projected local datasets, which are used by the central server to generate the global coded dataset. In each iteration of the training process, the central server aggregates gradients received from the non-stragglers and gradient computed from the global coded dataset with fixed aggregation weights. Compared with previous methods, the advantages of SCFL lie in ensuring that the global gradient applied by the central server to update the global model is an unbiased estimation of the true gradient, and providing the trade-off between the privacy performance and the learning performance analytically.  However, in SCFL, when the central server aggregates the gradients, it consistently adopts pre-specified and fixed weights across iterations, which neglects the generation process of the global coded dataset and the dynamic nature of the trained model over iterations. This oversight results in an unnecessary performance loss, preventing SCFL from achieving the optimal performance in terms of privacy and learning.
 
\begin{figure}[h]
    \centering
    \includegraphics[width=0.8\linewidth]{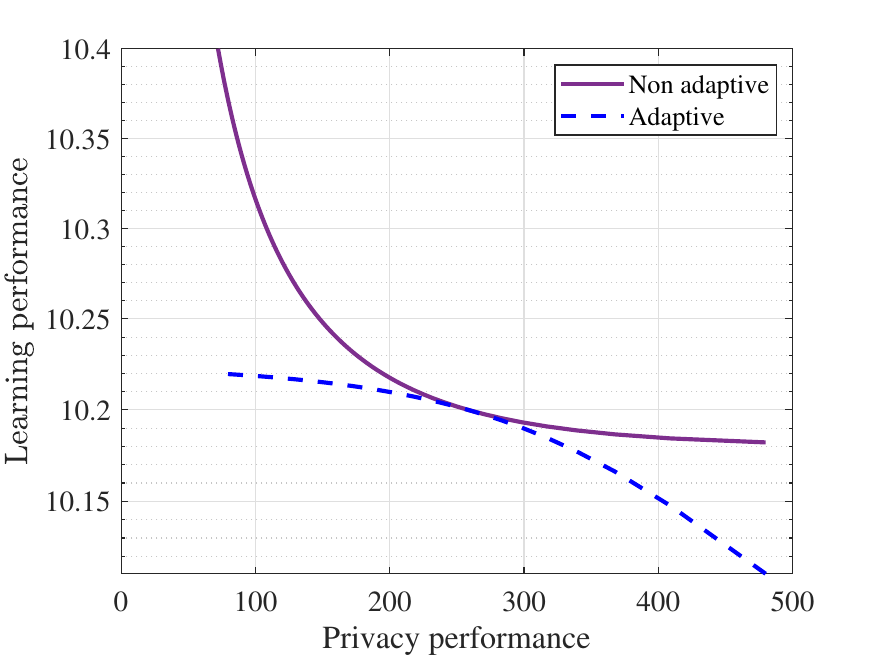}
    \caption{The trade-off between the learning and privacy performance under non-adaptive aggregation weights and adaptive aggregation weights, where the non-adaptive aggregation weight of the gradient computed by the server based on the global coded dataset is fixed at 0.1.}
    \label{fig: illus}
\end{figure}

For FL under requirements of privacy preservation and straggler mitigation, to overcome the drawbacks of existing CFL methods as mentioned above, we propose a new method named adaptive coded federated learning (ACFL). In ACFL, before training, each device uploads a local coded dataset, which is a transformed version of its local dataset with additive random noise. Upon receiving local coded datasets from the devices, the central server generates the global coded dataset. Then, during each training iteration, each non-straggler device computes local gradient based on the current global model and its local dataset, subsequently transmitting this gradient to the central server. After receiving local gradients from the non-straggler devices, the central server estimates the global gradient by aggregating the received gradients and the gradient computed from the global coded dataset, applying aggregation weights determined by an adaptive policy. Under this policy, the performance in terms of privacy and learning is optimized, where the learning performance is described by convergence analysis and the privacy performance in sharing local coded datasets with the server is measured by mutual information differential privacy (MI-DP).  As an illustration, Fig.~\ref{fig: illus} presents the relationship between learning performance and privacy performance based on our theoretical analysis, comparing adaptive aggregation weights with non-adaptive (i.e., fixed) weights. The $x$-axis represents privacy performance, where a smaller value indicates higher privacy level. The $y$-axis represents learning performance, where a smaller value indicates faster convergence (See Fig.~\ref{fig: the_per} in Section~\ref{parameter setting of ACFL} for more descriptions.).  It can be observed that by using adaptive aggregation weights determined by the adaptive policy, better learning performance is achieved for the same level of privacy performance, compared to the case of using fixed weights, as done in existing methods such as SCFL.   Finally, in order to verify the superiority of the proposed method, simulations results are provided, where we compare the performance of ACFL and the baseline methods. 

It is worth noting that although both ACFL and SCFL adopt the CFL framework, there are significant differences between them, which are outlined below:
\begin{enumerate}
    \item ACFL and SCFL employ different methods for encoding local datasets at the devices and generating the global coded dataset at the server. Unlike SCFL, which uses a separate projection matrix to transform local datasets for encoding, ACFL projects each local dataset using its own transpose. The advantage of this approach is that the transformed local datasets retain information equivalent to the raw datasets for gradient computation, leading to better learning performance at the same privacy level. This rationale will be explained in more detail in Section~\ref{the proposed methods} and further verified in Section~\ref{simulations}.
    \item Unlike SCFL, which consistently uses pre-specified and fixed weights across iterations, ACFL fully considers the generation process of the global coded dataset and the dynamic nature of the trained model over iterations and uses adaptive aggregation weights to aggregate gradients at the server for model updates. The use of adaptive aggregation weights ensures optimal performance in terms of both privacy and learning for the proposed method.
    \item Due to the very different implementations, the theoretical analysis presented in this paper is novel and distinct from that in SCFL, even though both methods use the same measure, MI-DP, to characterize the privacy level of the FL system in sharing local coded datasets with the server.
\end{enumerate}

Our contributions are summarized as follows:
\begin{enumerate}
    \item We propose a new ACFL method designed for FL problems under requirements of straggler mitigation and privacy preservation. This method develops an adaptive policy for adjusting the aggregation weights at the central server during gradient aggregation, aiming to achieve the optimal performance in terms of privacy and learning. 
    \item We analyze the trade-off between learning and privacy performance of the proposed method in sharing local coded datasets with the server.
    \item Through simulations, we demonstrate that the proposed method surpasses the baseline methods, which attains better learning performance under certain privacy levels. 
\end{enumerate}

The rest of this paper is structured as follows. In Section \ref{problem model}, we formulate the considered problem. In Section \ref{the proposed methods}, we propose our ACFL method. In Section \ref{performance analysis}, we analyze the performance of ACFL theoretically and describes the adaptive policy adopted by the proposed method. Simulation results are provided in Section \ref{simulations} to verify the superiority of the proposed method. Finally, we conclude this paper in Section \ref{conclusions}.

\section{Problem Model}
\label{problem model}
In the FL system, there is a central server and \(N\) devices, where each device \(i\) owns a local dataset consisting of \({{\mathbf{X}}^{(i)}} \in \mathbb{R}^{m_i \times d}\) and \({{\mathbf{Y}}^{(i)}} \in \mathbb{R}^{m_i \times o}\), \(\forall i\). Here, \({{\mathbf{X}}^{(i)}}\) contains \(m_i\) data samples of dimension \(d\), and \({{\mathbf{Y}}^{(i)}}\) represents the corresponding labels of dimension \(o\). We also assume that, for each local dataset, the number of data samples \(m_i\) is greater than \(d\) and \({{\mathbf{X}}^{(i)}} \in \mathbb{R}^{m_i \times d}\) is full column rank \cite{showkatbakhsh2018privacy}. Under the coordination of the central server, the aim is to minimize the total training loss $f\left( \mathbf{W} \right)$ by finding the solution to the following linear regression problem \cite{sun2023stochastic}:
\begin{align}
\label{model1}
{{\mathbf{W}}^*} = \arg \min_{\mathbf{W} \in \mathbb{R}^{d \times o}} f\left( \mathbf{W} \right) = \sum_{i = 1}^N \frac{1}{2}\left\| {{\mathbf{X}}^{(i)}\mathbf{W} - {{\mathbf{Y}}^{(i)}}} \right\|_F^2,
\end{align}
where \(\mathbf{W} \in \mathbb{R}^{d \times o}\) is the model parameter matrix, and \(\left\| \cdot \right\|_F\) denotes the Frobenius norm. Here, we focus on linear regression problems under the FL setting for two reasons. First, linear regression is a very important machine learning model that has attracted significant attention recently due to its ability to make scientific and reliable predictions. As a statistical procedure established a long time ago, its properties are well understood, enabling the rapid training of linear regression models in different applications \cite{maulud2020review, montgomery2021introduction, wen2021great, liu2024efficient}.  For instance, in agricultural applications, linear regression models can be used to predict rice cultivation timing based on weather data. In medical applications, simulated scenarios of traditional Chinese bone-setting techniques provide valuable data for analysis. By applying linear regression, one can establish relationships between measurable features of the manipulation process and the underlying conceptual knowledge, enabling the extraction of abstract insights from observable physical actions \cite{maulud2020review}.   Second, nonlinear problems can be transformed into linear regression problems through techniques such as high-order approximations or distributed kernel embedding \cite{prakash2020coded}, demonstrating the broad applicability of linear regression models.  In addition, many classification problems can be reformulated as linear regression problems as well. For instance, binary classification tasks can be approached using linear regression by encoding class labels numerically, e.g., 0 and 1, and interpreting the regression output as a probability estimate.

For the problem introduced above under the FL setting, the training process involves multiple iterations. In each iteration, the devices compute local gradients based on their local datasets and transmit to the central server. After aggregating the received gradients, the central server updates the global model and broadcasts the model parameters to all devices. Over the iterations, due to various incidents, some devices may become stragglers, which are temporarily unresponsive, fail to implement computations, and do not send any messages to the central server, unlike the normally functioning devices \cite{tandon2017gradient}. It is assumed that the probability of each device being a straggler is $p$ in each iteration, and the straggler behavior of the devices is independent across iterations and among different devices \cite{bitar2020stochastic}.  This assumption corresponds to practical scenarios in which devices are equipped with hardware or resources of similar quality, and stragglers arise due to temporary computation or communication issues, such as transient hardware errors \cite{harlap2016addressing}. These types of errors are typically independent across devices and over time, and occur with roughly equal likelihood on each device \cite{harlap2016addressing}.  Additionally, the probability of becoming a straggler is assumed to remain constant, which is reasonable when the behavior of devices is relatively stationary during the training process\footnote{{The methodology and analysis presented in this paper can be extended to scenarios in which straggler behavior is heterogeneous or correlated across devices, as well as to cases where the straggler behavior is time-varying or exhibits temporal correlation across iterations due to resource dependency. However, such extensions are beyond the primary focus of this work, which are left for future extensions.}} \cite{tandon2017gradient}.

 The presence of stragglers may degrade learning performance or even lead to poor convergence. As an illustrative example, Fig.~\ref{fig: straggler} shows the model update process with and without stragglers. When there are no stragglers, the global model is consistently updated in the direction of steepest descent. In contrast, in the presence of stragglers, missing gradients from the stragglers introduce gradient misalignment, which destabilizes the optimization trajectory and slows convergence. To be more specific, as shown in Fig.~\ref{fig: straggler}, after the same number of iterations starting from the same initial model, the optimal model can be reached in the absence of stragglers, where the training loss is minimized. However, in the presence of stragglers, the resulting model remains far from the optimum, indicating that only a suboptimal solution is obtained. In this case, additional iterations are required to reach the optimal model, compared to the scenario without stragglers.

\begin{figure}[h]
    \centering
    \includegraphics[width=\linewidth]{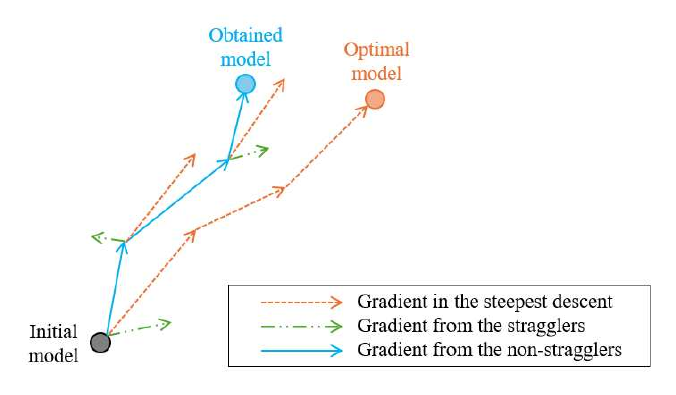}
    \caption{{An illustrative example showing the effects of the stragglers.}}
    \label{fig: straggler}
\end{figure}


To deal with the stragglers in FL, the CFL paradigm can be employed. In CFL, the local dataset \(\left( {\mathbf{X}^{(i)},\mathbf{Y}^{(i)}} \right)\) is encoded into \(\left( {\mathbf{H}_X^{(i)},\mathbf{H}_Y^{(i)}} \right)\) and transmitted to the central server by device \(i\) before training begins, \(\mathbf{H}_X^{(i)} \in \mathbb{R}^{n_X^{(i)} \times d}\), \(\mathbf{H}_Y^{(i)} \in \mathbb{R}^{n_Y^{(i)} \times o}\), $\forall i$. The central server utilizes the local coded datasets to generate a global coded dataset, denoted as \(\left( \mathbf{\tilde H}_X, \mathbf{\tilde H}_Y \right)\), instead of retaining all the received coded datasets to save its storage \cite{sun2023stochastic}. Subsequently, the training process begins, involving multiple iterations. In each iteration, each non-straggler device computes the gradient based on its local dataset and transmits it to the central server. Upon receiving these messages from the non-stragglers, the central server updates the current model parameter matrix using both the received gradients and the gradient computed from the global coded dataset \cite{sun2023stochastic}. At the end of each iteration, the central server broadcasts the updated model parameters to all devices. In CFL, the core issue is to design the method of generating global coded dataset before training and aggregating gradients at the central server during the training process, so that we can achieve a better performance in terms of privacy and learning. A recently proposed method within the CFL framework is known as SCFL \cite{sun2023stochastic}. However, in SCFL, consistently using fixed and pre-specified weights for aggregating gradients without considering the generation process of the global coded dataset and the dynamic nature of the trained model over iterations results in an unnecessary performance loss. Based on this, we aim to overcome the limitations of existing work and achieve better performance in terms of both privacy and learning, while fully exploiting the advantages of the CFL paradigm.

\section{The Proposed Method}
\label{the proposed methods}

In this section, we propose the ACFL method to address the problem outlined in Section \ref{problem model}. ACFL is implemented in two distinct stages. Initially, in the first stage prior to the commencement of training, each device transmits a coded dataset to the central server, which is derived by transforming its local dataset and adding noise. The central server then uses these local coded datasets to generate a global coded dataset. Subsequently, in the second stage, the training process involves multiple iterations. In each iteration, the central server aggregates gradient computed from the global coded dataset and the received gradients from non-straggler devices, where an adaptive policy to determine the aggregation weights is applied.
\subsection{The first stage: before training}
\label{subsection before training}
Before the training starts, each device uploads a coded dataset to the central server, which is a common practice under the framework of CFL. In the literature, methods commonly used to generate local coded datasets include matrix projection and noise injection \cite{dhakal2019coded,sun2023stochastic,showkatbakhsh2018privacy}. Matrix projection reduces the volume of data transmission by lowering the dimension of the projected data compared to the raw datasets. Meanwhile, noise injection enables privacy enhancement at adjustable levels. Inspired by this, in ACFL, the devices encode the local datasets by adding noise to a transformed version of the local datasets. To be more specific, device $i$ uploads coded dataset  \(\left( {\mathbf{H}_X^{(i)},\mathbf{H}_Y^{(i)}} \right)\) given as 
\begin{align}
    \label{local coded dataset}
    \mathbf{H}_X^{(i)} = {{\mathbf{X}}^{\left( i \right)T}}{{\mathbf{X}}^{\left( i \right)}} + {\mathbf{N}}_1^{\left( i \right)},
\end{align}
and
\begin{align}
    \label{coded label}
    \mathbf{H}_Y^{(i)} = {{\mathbf{X}}^{\left( i \right)T}}{{\mathbf{Y}}^{\left( i \right)}} + {\mathbf{N}}_2^{\left( i \right)},
\end{align}
to the central server, where ${\mathbf{N}}_1^{\left( i \right)} \in {\mathbb{R}^{d \times d}}$ and ${\mathbf{N}}_2^{\left( i \right)} \in {\mathbb{R}^{d \times o}}$ denote the additive noise. In ${\mathbf{N}}_1^{\left( i \right)}$, all the elements are independent and identically distributed (i.i.d.) random variables drawn from Gaussian distribution \(\mathcal{N}\left( {0,\sigma _1^2} \right)\). Similarly, all the elements in ${\mathbf{N}}_2^{\left( i \right)}$ are i.i.d. random variables following distribution \(\mathcal{N}\left( {0,\sigma _2^2} \right)\). Differing from existing methods that adopt random matrix projection \cite{dhakal2019coded,sun2023stochastic,showkatbakhsh2018privacy}, ACFL projects each local dataset using its own transpose. This approach still reduces the volume of transmitted data compared to sending the raw datasets, since $d < m_i$ always holds. Furthermore, as will be detailed in Section~\ref{subsection training iterations}, $\mathbf{X}^{(i)T}\mathbf{X}^{(i)}$ and $\mathbf{X}^{(i)T}\mathbf{Y}^{(i)}$ contain information equivalent to that in the raw datasets for computing gradients in linear regression problems. This implies the efficiency of projecting the local datasets as described in (\ref{local coded dataset}) and (\ref{coded label}).  Furthermore, it is worth emphasizing that the computation of $\mathbf{X}^{(i)T}\mathbf{X}^{(i)}$ and $\mathbf{X}^{(i)T}\mathbf{Y}^{(i)}$ during the encoding operation on the devices incurs almost no additional computational load. This is because, in the absence of the encoding operation, $\mathbf{X}^{(i)T}\mathbf{X}^{(i)}$ and $\mathbf{X}^{(i)T}\mathbf{Y}^{(i)}$ would still need to be computed later by each device to obtain the local gradients during the second stage of training iterations, as described in Section~\ref{subsection training iterations}.

Upon receiving local coded datasets, the central server generates a global coded dataset as
\begin{align}
    \label{global coded dataset}
    {{{\mathbf{\tilde H}}}_X} = \sum\limits_{i = 1}^N {{\mathbf{H}}_X^{(i)}},
    {{{\mathbf{\tilde H}}}_Y} = \sum\limits_{i = 1}^N {{\mathbf{H}}_Y^{(i)}}. 
\end{align}
In this way, the storage at the central server can be saved significantly compared with storing all received coded datasets directly.

\subsection{The second stage: training iterations}
\label{subsection training iterations}
In the $t$-th iteration, let us denote the current model parameter matrix as ${{\mathbf{W}}_t}$. The following random variables are adopted to represent the straggler behaviour of the devices:
\begin{align}
\label{Iit}
I_t^{\left( i \right)} = 
\begin{cases} 
1, & \text{if device } i\text{ is not a straggler}, \\
0, & \text{if device } i\text{ is a straggler}.
\end{cases}
\end{align}
According to (\ref{Iit}), \(\left\{ {I_t^{\left( i \right)},\forall t,\forall i} \right\}\) are i.i.d. Bernoulli random variables, whose probability mass function can be expressed as 
\begin{align}
\label{Iit_bernoulli}
\Pr \left( {I_t^{\left( i \right)} = 1} \right) = 1 - p, \Pr \left( {I_t^{\left( i \right)} = 0} \right) = p.
\end{align}
If device $i$ is not a straggler in the current iteration, it computes the local gradient as follows:
\begin{align}
    \label{local model update}
    {\mathbf{G}}_t^{\left( i \right)} = {{\mathbf{X}}^{\left( i \right)T}}\left[ {{{\mathbf{X}}^{\left( i \right)}}{{\mathbf{W}}_t} - {{\mathbf{Y}}^{\left( i \right)}}} \right],
\end{align}
and then transmits ${\mathbf{G}}_t^{\left( i \right)}$ to the central server. From (\ref{local model update}), we can clearly see that $\mathbf{X}^{(i)T}\mathbf{X}^{(i)}$ and $\mathbf{X}^{(i)T}\mathbf{Y}^{(i)}$ contain information equivalent to that in the raw datasets for computing gradients in the considered problem, which implies that it is reasonable to project the local datasets as described in (\ref{local coded dataset}) and (\ref{coded label}) to generate the coded datasets. At the same time, the central server computes a gradient based on the global coded dataset to obtain:
\begin{align}
    \label{server update}
    {\mathbf{G}}_t^S = {{{\mathbf{\tilde H}}}_X}{{\mathbf{W}}_t} - {{{\mathbf{\tilde H}}}_Y}.
\end{align}
Subsequently, with the received gradients as shown in (\ref{local model update}) and the gradient computed as (\ref{server update}), the server derives the global gradient as 
\begin{align}
    \label{global model update}
    {\mathbf{G}}_t^{All} = {\alpha _t}{\mathbf{G}}_t^S + \frac{{1 - {\alpha _t}}}{{1 - p}}\sum\limits_{i = 1}^N {{\mathbf{G}}_t^{\left( i \right)}I_t^{\left( i \right)}},
\end{align}
where $\alpha _t \in [0,1]$ represents the adaptive aggregation weight. For the proposed method, we will introduce an adaptive policy for adjusting the value of $\alpha _t$ over iterations to attain the optimal performance in terms of privacy and learning after conducting performance analysis in the next section. With the global gradient provided in (\ref{global model update}), the central server updates the model parameter matrix as
\begin{align}
    \label{update model}
    {{\mathbf{W}}_{t + 1}} = {{\mathbf{W}}_t} - \eta_t {\mathbf{G}}_t^{All},
\end{align}
and broadcasts ${{\mathbf{W}}_{t + 1}}$ to all devices, which ends the $t$-th iteration. In (\ref{update model}), $\eta_t$ denotes the learning rate. The implementation of the proposed method is shown as Fig.~\ref{fig: framework}.  
\begin{figure}[h]
    \centering
    \includegraphics[width=\linewidth]{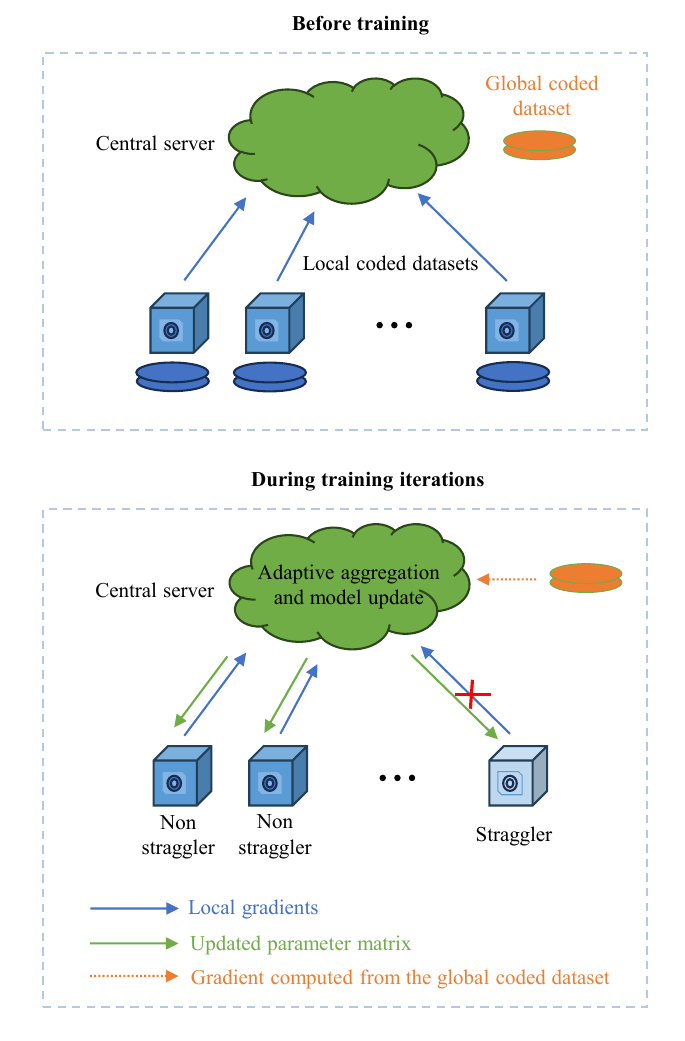}
    \caption{The implementation of ACFL.}
    \label{fig: framework}
\end{figure}
\begin{remark}
    \label{remark 00}
    For the proposed method, the overall communication overhead during $T$ iterations, measured by the number of uploaded bits from the devices to the central server, can be expressed as
    \begin{align}
        \label{ACFL communication overhead}
        {\psi ^{(T)}} = \psi _1^{(T)} + \psi _2^{(T)},
    \end{align}
    where $\psi _1^{(T)}$ is the communication overhead in the first stage induced by transmission of local coded datasets and $\psi _2^{(T)}$ is the communication overhead in the second stage resulted from uploading gradients from the devices to the central server. Based on that, we can derive $\psi _1^{(T)}$ and $\psi _2^{(T)}$ as follows:
    \begin{align}
        \label{psi}
        \psi _1^{(T)} &= \phi\left( {{d^2} + od} \right)N,\\
        \psi _2^{(T)} &= \phi od\sum\limits_{i = 1}^N {\sum\limits_{t = 0}^T {I_t^{\left( i \right)}} },
    \end{align}
    where $\phi$ denotes the number of bits required for the transmission of a real number. In many practical applications, the FL process involves a large number of iterations with \(T \gg \frac{{d + o}}{{o\left( {1 - p} \right)}}\), which can be verified under our simulation settings in  Section~\ref{simulations}, indicating $\psi _2^{(T)} \approx \phi odN(1-p)\left( {T + 1} \right) \gg \psi _1^{(T)}$. Hence, it holds that ${\psi ^{(T)}} \approx \psi _2^{(T)}$. Based on this, the additional communication overhead introduced by the first stage of the proposed method is negligible compared to the overall communication cost of the training process.

        Note that for conventional FL involving only training iterations \cite{mcmahan2017communication}, the communication overhead arises from uploading gradients from the devices to the central server in each iteration. In this context, the per-iteration communication overhead of conventional FL and ACFL is equivalent. Given that the additional communication overhead introduced by the first stage of ACFL is negligible compared to the communication cost incurred during training iterations, both conventional FL and ACFL incur approximately the same overall communication burden after the same number of iterations.     
\end{remark}

\section{Performance Analysis and The Adaptive Policy}
\label{performance analysis}
In this section, we first analyze the performance of the proposed method with arbitrary aggregation weights from two aspects: privacy performance and learning performance. Based on that, to attain the optimal performance in terms of privacy and learning, the adaptive policy of determining the value of $\alpha_t$ would be specified accordingly. 
 Note that, in this paper, we focus on analyzing the privacy performance during the sharing of local coded datasets with the server, rather than the transmission of local gradients during training iterations, similar as \cite{sun2023stochastic,prakash2020coded}. This focus is motivated by the observation that, compared with the typical FL framework, the additional privacy concerns introduced by the proposed method arise only in the first stage, where local coded datasets are shared with the server. In contrast, the privacy concerns associated with sharing local gradients exist universally across all FL scenarios. Based on that, it is sufficient to analyze the additional privacy concerns in the ACFL framework to evaluate the advantages and limitations of the proposed method.
Additionally, numerous works in the FL literature address privacy concerns related to sharing local gradients with the server \cite{el2022differential}. These techniques can be seamlessly incorporated into the proposed method to enhance privacy during the second stage. However, this is beyond the primary scope of this work. In light of this, in this paper, when we refer to the privacy performance of ACFL, we specifically mean the privacy performance of ACFL in sharing local coded datasets with the server.

First, let us provide some useful assumptions, which have been widely and reasonably adopted in the literature.
\begin{assumption}
\label{assumption absolute values in x and y}
It is assumed that the absolute values of all elements in \({{\mathbf{X}}^{(i)}}\) and \({{\mathbf{Y}}^{(i)}}\) are no larger than 1, \(\forall i\) \cite{showkatbakhsh2018privacy}.
\end{assumption}
\begin{assumption}
    \label{assumption bounded gradients}
The local gradients computed as (\ref{local model update}) are all bounded as \cite{sun2023stochastic}
\begin{align}
    \label{bounded gradients}
    \left\| {{\mathbf{G}}_t^{\left( i \right)}} \right\|_F^2 \leq {\beta ^2},\forall i.
\end{align}
\end{assumption}
\begin{assumption}
    \label{assumption bounded parameter}
The parameter matrix is bounded as follows \cite{sun2023stochastic}:
    \begin{align}
        \label{parameter bound}
  \left\| {{{\mathbf{W}}_t}} \right\|_F^2 \leq {C^2}.
  \end{align}
\end{assumption}

\subsection{Privacy performance of ACFL}
\label{section Privacy performance of ACFL}

 To evaluate the privacy performance of the proposed method in sharing local coded datasets with the server,  the metric of MI-DP is adopted, which is stronger than the conventional differential privacy metric \cite{cuff2016differential, li2021privacy}. Let us provide the definition of MI-DP as follows \cite{showkatbakhsh2018privacy}:
\begin{definition}[$\epsilon$-MI-DP]
\label{def-mi-dp}
The coded dataset $q\left( {\mathbf{Z}} \right) = {\mathbf{\tilde Z}} \in {\mathbb{R}^{{a_2} \times {b_2}}}$ obtained from the original dataset ${\mathbf{Z}} \in {\mathbb{R}^{{a_1} \times {b_1}}}$ with encoding operation $q(\cdot)$ satisfies $\epsilon$-MI-DP if
\begin{equation}
\label{midp}
\sup_{j,k,\Gamma({\mathbf{Z}})} I(Z_{j,k};q({\mathbf{Z}})\|\mathbf{Z}^{-j,k}) \leq \epsilon,
\end{equation}
where $Z_{j,k}$ denotes the $(j,k)$-th element in ${\mathbf{Z}}$, $\mathbf{Z}^{-j,k}$ is the set of all elements in ${\mathbf{Z}}$ excluding $Z_{j,k}$, $\Gamma({\mathbf{Z}})$ denotes the distribution of ${\mathbf{Z}}$, and $I(\cdot||\cdot)$ is the conditional mutual information. 

In this definition, $\epsilon$ is an upper bound of the maximum privacy leakage of each element in the original dataset ${\mathbf{Z}}$ given the coded dataset $q\left( {\mathbf{Z}} \right)$. A decreasing value of $\epsilon$ indicates enhanced privacy performance. 
\end{definition}

Based on Definition \ref{def-mi-dp}, the privacy performance of ACFL is characterized in the following Theorem. 
\begin{theorem}[Privacy performance of ACFL]
\label{Privacy performance of ACFL}
With respect to device $i$, $\forall i$, the local coded dataset satisfies $\epsilon$-MI-DP, where
\begin{align}
    \label{eps}
    \epsilon=\left( {d - \frac{1}{2}} \right)\log \frac{{1 + \sigma _1^2}}{{\sigma _1^2}} + \frac{o}{2}\log \frac{{1 + \sigma _2^2}}{{\sigma _2^2}}.
\end{align}

\end{theorem}
\begin{proof}
        Please see Appendix~\ref{appendix privacy}.
\end{proof}
\begin{remark}
    \label{remark 1}
From Theorem \ref{Privacy performance of ACFL}, we observe that \(\epsilon\) is a decreasing function with respect to \(\sigma_1^2\) and \(\sigma_2^2\). By Definition \ref{def-mi-dp}, the privacy performance of ACFL is enhanced by increasing the strength of the additive noise, that is, by increasing the values of \(\sigma_1^2\) and \(\sigma_2^2\), aligning with our intuition. Furthermore, it holds that 
\begin{align}
    \label{asy}
    \lim_{\sigma_1^2, \sigma_2^2 \to \infty} \epsilon = 0,
\end{align}
which signifies that perfect privacy can be asymptotically achieved when the variances of the additive noise approach infinity.
\end{remark}
\begin{remark}
In the typical FL setting, gathering the training data of all devices at the server is challenging due to two main reasons: (1) sharing raw training data directly with the server raises privacy concerns, especially when the data is sensitive, and (2) the communication overhead of transmitting local datasets to the server is significant.

In contrast, the proposed ACFL framework addresses both issues by allowing devices to share local coded datasets with the server. First, privacy concerns are mitigated by transforming the original datasets with additive noise before transmission. This ensures that the privacy leakage caused by transmitting data can be well controlled, as shown in Theorem~\ref{Privacy performance of ACFL} and Remark~\ref{remark 1}. Second, the communication overhead is significantly reduced by transmitting local coded datasets instead of raw datasets, as explained in Section~\ref{subsection before training}. 
\end{remark}
\subsection{Learning performance of ACFL}
\label{section learning performance of ACFL}

Let us first provide two lemmas, which facilitate the subsequent derivations of Theorem 2 describing the learning performance of ACFL under arbitrary values of aggregation weights.  

\begin{lemma}
\label{unbiased}
The global gradient applied by the central server denoted as (\ref{global model update}) is an unbiased estimation of the true gradient given as
\begin{align}
    \label{true gradient}
    {\mathbf{G}}_t^{\text{True}} = \sum\limits_{i = 1}^N {{\mathbf{G}}_t^{\left( i \right)}},
\end{align}
where \({{\mathbf{G}}_t^{\left( i \right)}}\) is provided as (\ref{local model update}).
\end{lemma}
\begin{proof}
According to (\ref{local coded dataset}), (\ref{coded label}), (\ref{global coded dataset}) and (\ref{server update}), we have
\begin{align}
    \label{global_update_unbiased}
    \mathbb{E}\left( {\left. {{\mathbf{G}}_t^S} \right|{{\mathbf{W}}_t}} \right) = \mathbb{E}\left( {{{{\mathbf{\tilde H}}}_X}} \right){{\mathbf{W}}_t} - \mathbb{E}\left( {{{{\mathbf{\tilde H}}}_Y}} \right) = {\mathbf{G}}_t^{\text{True}}.
\end{align}
Further, based on (\ref{Iit_bernoulli}), (\ref{global model update}) and (\ref{global_update_unbiased}), we have 
\begin{align}
    \label{unbiased1}
    \mathbb{E}\left( {\left. {{\mathbf{G}}_t^{All}} \right|{{\mathbf{W}}_t}} \right) = \mathbb{E}\left[ {\left. {{\alpha _t}{\mathbf{G}}_t^S + \frac{{1 - {\alpha _t}}}{{1 - p}}\sum\limits_{i = 1}^N {{\mathbf{G}}_t^{\left( i \right)}I_t^{\left( i \right)}} } \right|{{\mathbf{W}}_t}} \right] \nonumber\\
    = {\alpha _t}{\mathbf{G}}_t^{\text{True}} + \left( {1 - {\alpha _t}} \right)\sum\limits_{i = 1}^N {{\mathbf{G}}_t^{\left( i \right)}}  = {\mathbf{G}}_t^{\text{True}},
\end{align}
which completes the proof. 
\end{proof}

\begin{lemma}
\label{bounded variance}
The global gradient in (\ref{global model update}) is bounded as
\begin{align}
    \label{model update bound}
  &\mathbb{E}\left[ {\left\| {{\mathbf{G}}_t^{All}} \right\|_F^2} \right] \nonumber \\
   \leq& \left[ {\alpha _t^2p + \left( {1 - p} \right){{\left( {{\alpha _t} + \frac{{1 - {\alpha _t}}}{{1 - p}}} \right)}^2} + N - 1} \right]N{\beta ^2} \nonumber \\
   &+ \alpha _t^2Nd\sigma _1^2{C^2} + \alpha _t^2N\sigma _2^2od. 
\end{align}
\end{lemma}
\begin{proof}
        Please see Appendix~\ref{appendix bounded variance}. 
\end{proof}

\begin{theorem}[Convergence performance of ACFL under arbitrary values of aggregation weights]
\label{theroem convergence performance}
Under Assumptions 1-3, if $\sum\limits_{i = 1}^N {{{\mathbf{X}}^{\left( i \right)T}}{{\mathbf{X}}^{\left( i \right)}}}  \geq \lambda {\mathbf{I}}$ with $\lambda>0$, by setting ${\eta _t} = {1 \mathord{\left/
 {\vphantom {1 {\left( {\lambda t} \right)}}} \right. \kern-\nulldelimiterspace} {\left( {\lambda t} \right)}}$, it holds that 
\begin{align}
    \label{convergence}
    \mathbb{E}\left( {\left\| {{{\mathbf{W}}_T} - {{\mathbf{W}}^*}} \right\|_F^2} \right) \leq \mathop {\sup }\limits_{0 \leq t \leq T} 4\frac{u\left( {{\sigma _1^2,\sigma _2^2,{\alpha _t}}} \right)}{{{\lambda ^2}T}},
\end{align}
where
\begin{align}
    \label{u}
   & u\left( {{\sigma _1^2,\sigma _2^2,{\alpha _t}}} \right)\nonumber\\
    \triangleq& \left[ {\alpha _t^2p + \left( {1 - p} \right){{\left( {{\alpha _t} + \frac{{1 - {\alpha _t}}}{{1 - p}}} \right)}^2} + N - 1} \right]N{\beta ^2}\nonumber\\
    &+ \alpha _t^2Nd\sigma _1^2{C^2} + \alpha _t^2N\sigma _2^2od.
\end{align}
\end{theorem}
\begin{proof}
    When $\sum\limits_{i = 1}^N {{{\mathbf{X}}^{\left( i \right)T}}{{\mathbf{X}}^{\left( i \right)}}}  \geq \lambda {\mathbf{I}}$ with $\lambda>0$, the linear regression problem is $\lambda$-strongly convex. Based on that, Theorem \ref{theroem convergence performance} can be easily proved by combining Lemma \ref{unbiased}, Lemma \ref{bounded variance}, and Lemma 1 in \cite{rakhlin2011making}.
\end{proof}
Note that the optimal solution for the considered linear regression problem is given as
\begin{align}
\label{optimal solution}
    \mathbf{W}^* = \left( \sum_{i=1}^N \mathbf{X}^{(i)T} \mathbf{X}^{(i)} \right)^{-1} \left( \sum_{i=1}^N \mathbf{X}^{(i)T} \mathbf{Y}^{(i)} \right),
\end{align}
which relies solely on the matrices ${{\mathbf{X}}^{\left( i \right)T}}{{\mathbf{X}}^{\left( i \right)}}$ and ${{\mathbf{X}}^{\left( i \right)T}}{{\mathbf{Y}}^{\left( i \right)}}$. 

Based on (\ref{optimal solution}), the server can obtain an approximate version of the exact optimal solution using the received local coded datasets $\mathbf{H}_X^{(i)}$ and $\mathbf{H}_Y^{(i)}$ in the first stage, as defined in (\ref{local coded dataset}) and (\ref{coded label}). In other words, the central server receives a noisy version of the matrices ${{\mathbf{X}}^{\left( i \right)T}}{{\mathbf{X}}^{\left( i \right)}}$ and ${{\mathbf{X}}^{\left( i \right)T}}{{\mathbf{Y}}^{\left( i \right)}}$ in the first stage and can determine an approximate solution immediately.
However, the error in this approximate solution can become arbitrarily large as the strength of the additive noise increases, which corresponds to an enhanced privacy performance of ACFL, as shown in Remark~\ref{remark 1}. When there is a stringent requirement for privacy performance, this approximate solution becomes inapplicable. Besides, in many applications, it is essential to acquire the exact solution rather than an approximate one. For instance, in medical data analysis, where the linear regression models are used to predict diseases, an exact solution ensures accurate diagnostics. Small errors in the solution could lead to incorrect diagnoses. 

In contrast, as observed from Theorem~\ref{theroem convergence performance}, the proposed method converges to the optimal solution exactly, providing an exact and reliable solution at any privacy level. Based on that, employing the proposed method to obtain an exact optimal solution is more practical than directly deriving an approximate solution from the local coded datasets.

\begin{remark}
    \label{remark 2}
    Note that the condition \(\sum_{i = 1}^N \mathbf{X}^{(i)T}\mathbf{X}^{(i)} \geq \lambda \mathbf{I}\) with \(\lambda>0\) in Theorem~\ref{theroem convergence performance} naturally holds for the considered problem formulated in Section~\ref{problem model}. This is because \(\mathbf{X}^{(i)} \in \mathbb{R}^{m_i \times d}\) is full column rank, and \(\mathbf{X}^{(i)T}\mathbf{X}^{(i)}\) is positive definite. In this case, \(\lambda>0\) can be any value that satisfies \(0<\lambda<\sum_{i = 1}^N \text{eig}_{\min} \left[ \mathbf{X}^{(i)T}\mathbf{X}^{(i)} \right]\), where \(\text{eig}_{\min}(\cdot)\) denotes the smallest eigenvalue of a matrix.

\end{remark}

\subsection{The adaptive policy}
\label{parameter setting of ACFL}
Based on Section \ref{section Privacy performance of ACFL} and Section \ref{section learning performance of ACFL}, let us now consider how to adaptively determine the value of $\alpha_t$ for ACFL to attain a better performance in terms of privacy and learning. According to Theorem \ref{Privacy performance of ACFL} and Theorem \ref{theroem convergence performance}, the problem of determining $\alpha _t$ in ACFL can be equivalently expressed as the following optimization problem:
\begin{align}
    \label{optimization problem}
  \mathop {\min }\limits_{{\alpha _t}} u\left( {\sigma _1^2,\sigma _2^2,{\alpha _t}} \right), s.t. 0\leq{\alpha _t}\leq1,
\end{align}
which aims to expedite the convergence of the learning process under certain privacy level with fixed variances of additive noise. Based on the optimization problem defined in (\ref{optimization problem}), we detail the adaptive policy for selecting the value of \(\alpha_t\) and analyze the privacy and learning performance under this policy in Theorem~\ref{theorem parameter setting} as follows.

\begin{theorem}
    \label{theorem parameter setting}
    For the proposed ACFL method, by setting
    \begin{align}
        \label{formula parameters}
        \alpha_t= \frac{{pN{\beta ^2}\frac{1}{{1 - p}}}}{{pN{\beta ^2}\frac{1}{{1 - p}} + Nd\sigma _1^2{C^2} + N\sigma _2^2od}},
    \end{align}
the optimal performance in terms of privacy and learning can be attained. In this case, under Assumptions 1-3, if $\sum\limits_{i = 1}^N {{{\mathbf{X}}^{\left( i \right)}}^T{{\mathbf{X}}^{\left( i \right)}}}  \geq \lambda {\mathbf{I}}$ with $\lambda>0$, by setting ${\eta _t} = {1 \mathord{\left/
 {\vphantom {1 {\left( {\lambda t} \right)}}} \right. \kern-\nulldelimiterspace} {\left( {\lambda t} \right)}}$, it holds that
\begin{align}
    \label{optimal performance}
    \mathbb{E}\left( {\left\| {{{\mathbf{W}}_T} - {{\mathbf{W}}^*}} \right\|_F^2} \right) \leqslant 4\frac{{\tilde u\left( {\sigma _1^2,\sigma _2^2} \right)}}{{{\lambda ^2}T}},
\end{align}
where
\begin{align}
    \label{u11}
 & \tilde u\left( {\sigma _1^2,\sigma _2^2} \right) \nonumber \\
   =&  - \frac{{{{\left( {pN{\beta ^2}\frac{1}{{1 - p}}} \right)}^2}}}{{N{\beta ^2}p\frac{1}{{1 - p}} + Nd\sigma _1^2{C^2} + N\sigma _2^2od}} \nonumber \\
   &+ N{\beta ^2}\frac{1}{{1 - p}} + N{\beta ^2}\left( {N - 1} \right).
\end{align}
With the learning performance provided as (\ref{optimal performance}) and the privacy performance described in Theorem~\ref{Privacy performance of ACFL}, the performance in terms of privacy and learning attained by the proposed method is characterized. 
\end{theorem}
\begin{proof}
    Problem (\ref{optimization problem}) can be equivalently expressed as the following optimization problem:
    \begin{align}
        \label{eq_pro}
        \mathop {\min }\limits_{{\alpha _t}} \hat u\left( {{\alpha _t}} \right),s.t.0 \leqslant {\alpha _t} \leqslant 1,
    \end{align}
where 
    \begin{align}
        \label{u_alpha}
        \hat u\left( {{\alpha _t}} \right) \triangleq& \alpha _t^2\left( {N{\beta ^2}p\frac{1}{{1 - p}} + Nd\sigma _1^2{C^2} + N\sigma _2^2od} \right) \nonumber\\
        &- \left( {2pN{\beta ^2}\frac{1}{{1 - p}}} \right){\alpha _t}.
    \end{align}
We note that $\frac{{pN{\beta ^2}\frac{1}{{1 - p}}}}{{ {pN{\beta ^2}\frac{1}{{1 - p}} + Nd\sigma _1^2{C^2} + N\sigma _2^2od} }} < 1$ always holds, which indicates that 
\begin{align}
    \label{optimal_alpha_sigma}
    &\mathop {\arg \min }\limits_{{\alpha _t}} \hat u\left( {{\alpha _t}} \right) \nonumber\\
    &= \frac{{pN{\beta ^2}\frac{1}{{1 - p}}}}{{pN{\beta ^2}\frac{1}{{1 - p}} + Nd\sigma _1^2{C^2} + N\sigma _2^2od}}.
\end{align}
Substituting (\ref{optimal_alpha_sigma}) into (\ref{u}) yields (\ref{optimal performance}), which completes the proof.
\end{proof}
    From Theorem~\ref{theorem parameter setting}, we observe that the aggregation weight corresponding to the gradient computed by the server based on the global coded dataset decreases as the strength of the additive noise increases. This aligns with our intuition that when stronger noise is added to the local coded datasets, the gradient derived from the global coded dataset becomes less accurate and contains more noise, making it less reliable for updating the global model. Additionally, when the straggler probability $p$ increases, devices are more likely to become stragglers, leading to a greater loss of information. With more missing information from the stragglers, the local gradients from the non-stragglers deviate further from the true gradient, making it more important to leverage the information provided by the server to compensate for the missing contributions. As a result, the aggregation weight of the gradient computed by the server based on the global coded dataset increases.
 
\begin{remark}
    According to Theorem~\ref{Privacy performance of ACFL} and Theorem~\ref{theorem parameter setting}, the learning performance of ACFL improves as the variances of the additive noise decrease, whereas the privacy performance benefits from an increase in these variances. This highlights the intrinsic trade-off between privacy and learning performance. Practically, the variances of the additive noise can be meticulously calibrated to achieve a desired trade-off, aligning with specific practical requirements.
\end{remark}

To verify that ACFL with the adaptive policy achieves a better performance in terms of privacy and learning, Fig. \ref{fig: the_per} illustrates the relationship between the upper bound derived in (\ref{convergence}) and \(\epsilon\) specified in (\ref{eps}) with non-adaptive and fixed values of aggregation weights as well as adaptive aggregation weights determined by the designed adaptive policy, where we set $d=100, \sigma_1^2=\sigma^2_2, o=10, p=0.1, N=5, \beta=10, c=1, T=1000$ and $\lambda=1$. Here, a lower value of the upper bound implies superior learning performance, whereas a lower \(\epsilon\) value signifies enhanced privacy performance. The results clearly show that under the same privacy level, the upper bound in (\ref{convergence}) reaches its minimum with the adaptive policy. This outcome indicates that ACFL, when implemented with the adaptive policy, attains the optimal performance in terms of privacy and learning.

\begin{figure}[h]
    \centering
    \includegraphics[width=0.8\linewidth]{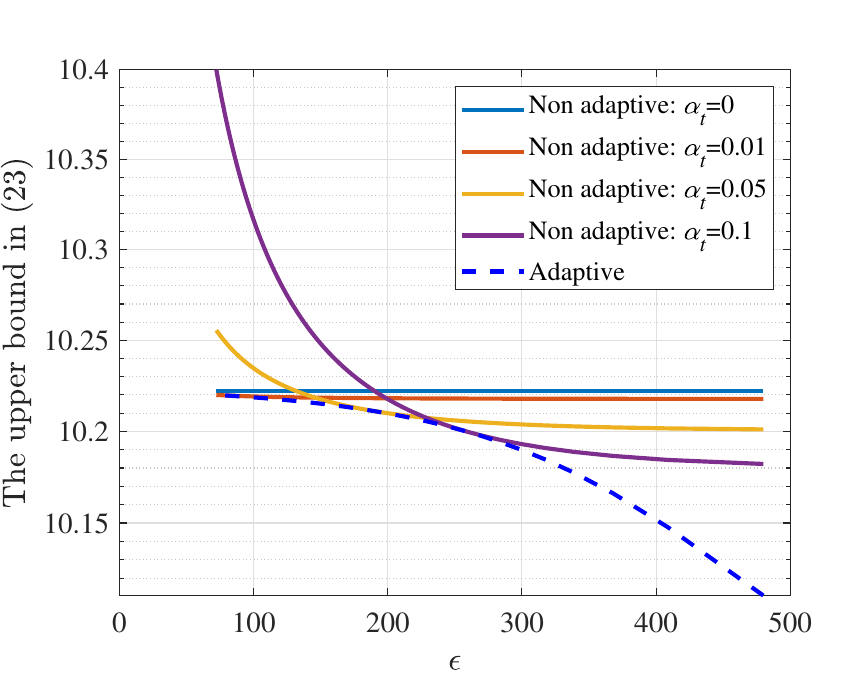}
    
    \caption{The trade-off between the learning and privacy performance is shown by depicting the upper bound in (\ref{convergence}) as a function of $\epsilon$ for ACFL under non-adaptive aggregation weights and adaptive aggregation weights.}
    \label{fig: the_per}
\end{figure}
In Theorem~\ref{theorem parameter setting}, we develop the adaptive policy for setting the aggregation weights in ACFL, which relies on the bounds specified in Assumptions~\ref{assumption bounded gradients} and~\ref{assumption bounded parameter}. Notably, obtaining the values of \(C\) and \(\beta\) beforehand can be challenging in certain scenarios. In such cases, a practical implementation of ACFL is estimating \(C\) and \(\beta\) over the iterations by the central server. Specifically, \(C\) and \(\beta\) are approximated by \(\hat{C}\) and \(\hat{\beta}\), respectively, with their expressions provided as follows:

\begin{align}
    \label{estimate beta2}
    \hat{C}^2 = \left\| {{{\mathbf{W}}_t}} \right\|_F^2, \hat{\beta}^2  = \frac{{\sum\limits_{i = 1}^N {\left\| {{\mathbf{G}}_t^{\left( i \right)}} \right\|_F^2I_t^{\left( i \right)}} }}{{\sum\limits_{i = 1}^N {I_t^{\left( i \right)}} }}.
\end{align}
Based on (\ref{estimate beta2}), the aggregation weights can be determined adaptively as
\begin{align}
        \label{near optimal parameters}
       {{\overset{\lower0.5em\hbox{$\smash{\scriptscriptstyle\smile}$}}{\alpha } }_t} = \frac{{p{{\hat \beta }^2}}}{{p{{\hat \beta }^2} + d\sigma _1^2{{\hat C}^2}\left( {1 - p} \right) + \sigma _2^2od\left( {1 - p} \right)}}.
    \end{align}
With (\ref{near optimal parameters}), the implementation of ACFL is presented in Algorithm \ref{alg:Algorithm ACFL}. 
    While the use of estimated bounds in (\ref{estimate beta2}) may lead to a slight degradation from theoretically optimal performance, our simulation results in Section~\ref{simulations} will demonstrate that the proposed method remains convergent and exhibits stable behavior across all tested scenarios. This indicates that the discrepancy between the theoretical analysis and practical implementation has a negligible impact on the stability and final performance of the algorithm. Moreover, the method consistently outperforms existing baselines, suggesting that the use of estimated bounds is sufficiently robust in practice.

From (\ref{estimate beta2}) and (\ref{near optimal parameters}), it can be seen that the computational load associated with computing the aggregation weights adaptively is $\mathcal{O}\left(odN(1 - p)\right)$. Based on this, the server-side weight computation time increases with the size of the model parameter matrix. However, this additional computational overhead does not necessarily indicate a delay in the overall algorithm, for the following reasons.
First, the server is typically equipped with sufficient computational resources to handle the additional load efficiently. As a result, the extra computation time incurred at the server is negligible when compared to the total computational cost of other components of the algorithm during training.
Second, the server can begin processing as soon as it starts receiving messages from the fastest non-straggler device. This enables the server-side computation to proceed in parallel with ongoing device-side computation and communication. In other words, the server does not need to wait until all non-straggler devices have completed their transmissions before beginning its computations.
This parallelism implies that the additional computation time introduced by computing the aggregation weights is negligible when compared to the baseline case in which fixed weights are employed by the server.
 
\begin{algorithm}
\caption{ACFL}\label{alg:Algorithm ACFL}
\textbf{Input:} Initial parameter matrix \({\mathbf{W}}_0\), learning rates $\eta_t$\\
\textbf{\textcolor{blue}{\underline{The first stage:}}} \textcolor{blue}{Device $i$ uploads coded dataset  \(\left( {\mathbf{H}_X^{(i)},\mathbf{H}_Y^{(i)}} \right)\) to the central server as (\ref{local coded dataset}) and (\ref{coded label}), $\forall i$.}\\
\textcolor{blue}{The central server generates a global coded dataset \(\left( {{{{\mathbf{\tilde H}}}_X},{{{\mathbf{\tilde H}}}_Y}} \right)\) as (\ref{global coded dataset}).}\\
\textbf{\textcolor{brown}{\underline{The second stage:}}}\\
\textcolor{brown}{\textbf{Initialize: }$t = 0$\\
 \While{\(t \leq T \) }{
 \textbf{In parallel for all devices} \(i \in \{1,2,...,N\}\)\textbf{:}\\
 \If{\(I_t^{\left( i \right)} = 1\)}{
Compute ${\mathbf{G}}_t^{\left( i \right)}$ as (\ref{local model update});\\
Transmits ${\mathbf{G}}_t^{\left( i \right)}$ to the central server;}
\textbf{The central server:}
Computes the gradient based on the global coded dataset to obtain $ {\mathbf{G}}_t^S$ as (\ref{server update});\\
Aggregates gradients to obtain ${\mathbf{G}}_t^{All}$ as (\ref{global model update}) with aggregation weights computed as (\ref{near optimal parameters});\\
Updates the model parameter matrix to obtain \({\mathbf{W}}_{t+1}\) as (\ref{update model});}
}
\end{algorithm}

\section{Simulations}
\label{simulations}
 In this section, we evaluate the performance of ACFL through simulations. Although ACFL is designed and analyzed for linear regression problems, it can also be applied to various nonlinear regression problems, considering that many nonlinear problems can be transformed into linear ones. To more extensively verify the effectiveness and superiority of ACFL, we consider both linear and nonlinear regression scenarios in our simulations.

\subsection{Linear regression scenarios} 
We consider an FL setting, where each device possesses a local dataset with $m_i = 100$ samples, $\forall i$.  The devices collaboratively train a linear regression model defined by a parameter matrix $\mathbf{W} \in \mathbb{R}^{d \times o}$, where $d = 10$ represents the input feature dimension of the data samples and $o = 10$ denotes the output dimension of the labels.  The elements in each local dataset ${{\mathbf{X}}^{(i)}}$ are i.i.d., drawn from a uniform distribution over the interval $[-1, 1]$. A true parameter matrix ${{\mathbf{W}}^{\text{True}}} \in {\mathbb{R}^{d \times o}}$ is generated with elements uniformly distributed within $[0, 1/30]$.  The labels ${{\mathbf{Y}}^{(i)}}$ for each device are generated according to ${{\mathbf{Y}}^{(i)}} = {{\mathbf{X}}^{(i)}}\left( {{{\mathbf{W}}^{\text{True}}} + i{{\mathbf{W}}^{\text{Shift}}}} \right)$ for all $i$, where ${{\mathbf{W}}^{\text{Shift}}}\in \mathbb{R}^{d \times o}$ has elements independently drawn from the uniform distribution over $[0, \sigma_s^2]$. A larger value of $\sigma_s^2$ corresponds to a greater degree of non-i.i.d.-ness in the training data across devices.  The initial parameter matrix ${\mathbf{W}}_0$ is randomly generated with elements drawn from a uniform distribution in the interval $[0, 1/30]$, independent of the true parameter matrix generation.

To evaluate the effectiveness of the adaptive policy in ACFL, we conduct a comparative analysis between ACFL and its non-adaptive counterpart (denoted as NA) with $N=100$ devices, which employs a fixed aggregation weight $\alpha_t = 0.5$, as suggested by \cite{sun2023stochastic}. In Fig. \ref{fig: compare_ACFL_NA}, we depict the training loss as a function of number of iterations for both ACFL and NA under different variances of additive noise\footnote{Note that under the same variances of additive noise, the privacy level is the same for both methods.}, where we set $\sigma_s^2=0$, $\eta_t=0.0001/t$ and we consider two distinct scenarios with varying straggler probabilities. As anticipated, an increase in noise variances, aimed at enhancing privacy, inversely affects the learning performance of both methods. This empirical observation aligns with our theoretical analysis, underscoring the intrinsic trade-off between privacy preservation and learning performance. Notably, the performance deterioration in ACFL remains relatively modest in comparison to NA upon increasing the noise, while there is a pronounced decline in the learning performance of NA. This discrepancy demonstrates the superior capability of ACFL to leverage gradients from both the central server and devices more effectively with the adaptive policy and validates the advantage of employing the adaptive policy in ACFL for mitigating the adverse effects of privacy-enhancing measures on learning performance. 

\begin{figure}[h]
    \centering
    \subfloat[]{%
        \includegraphics[width=0.8\linewidth]{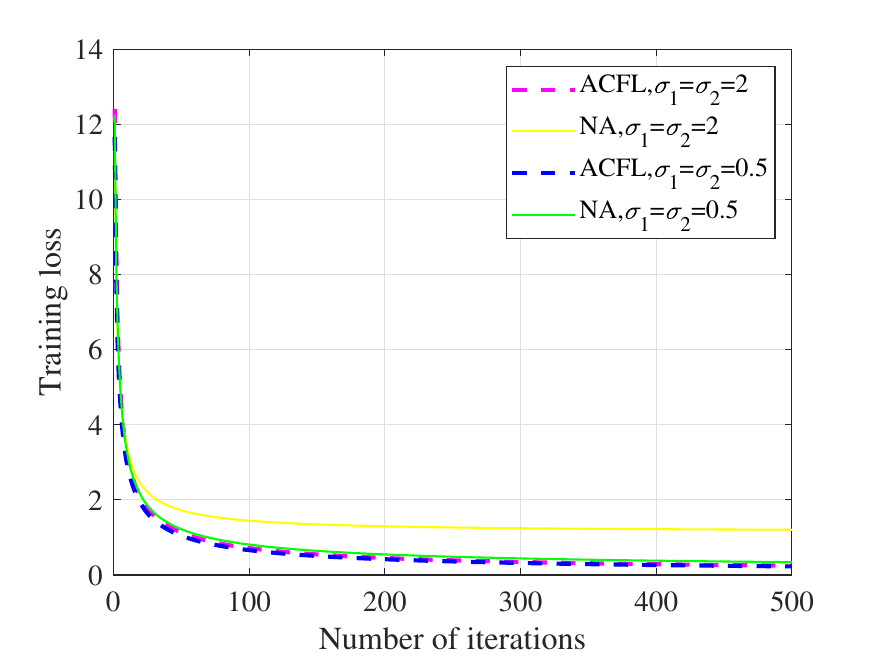}
        \label{fig: compare_ACFL_NA_p02}
    }

    \subfloat[]{%
        \includegraphics[width=0.8\linewidth]{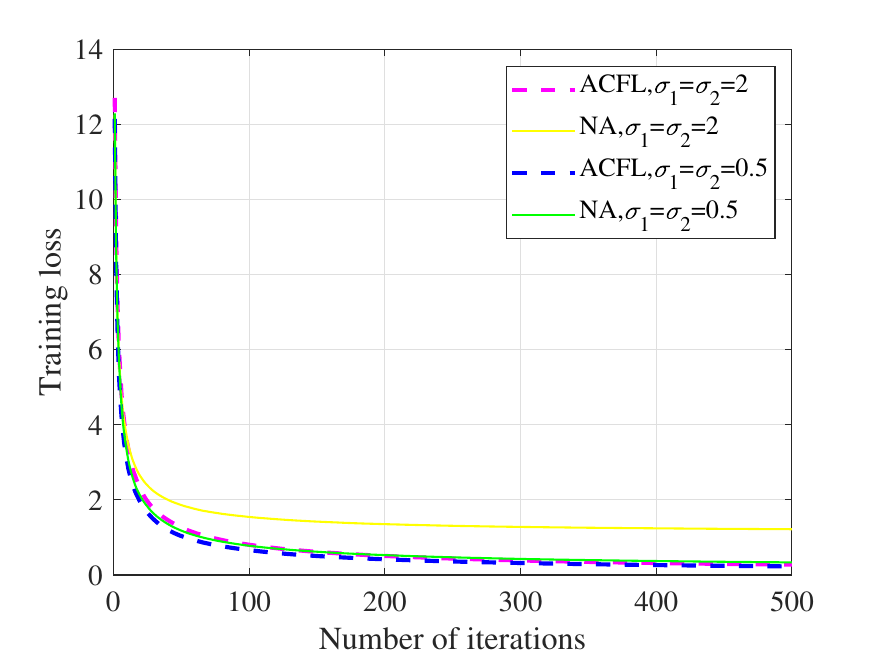}
        \label{fig: compare_ACFL_NA_p04}
    }
    \caption{Training loss as a function of number of iterations of ACFL and NA with varying noise variances. (a) The straggler probability $p=0.2$. (b) The straggler probability $p=0.4$. }
    \label{fig: compare_ACFL_NA}
\end{figure}

To further validate the efficacy of the proposed method, we compare its performance with the state-of-the-art method SCFL. Note that the difference between SCFL and ACFL lies in the way of encoding local datasets and the way of aggregating gradients at the central server, which has been detailed in Section~\ref{introduction}.  In Fig. \ref{fig: scfl}, we plot the training loss as a function of number of iterations for both methods under various privacy levels and varying degrees of non-i.i.d.-ness in the training data, i.e., under varying values of $\epsilon$ using the metric of $\epsilon$-MI-DP and varying values of $\sigma^2_s$.  The variances of the noise in both methods are calculated according to Theorem \ref{Privacy performance of ACFL} from this paper and Theorem 2 in \cite{sun2023stochastic}, ensuring required value of $\epsilon$ for both methods. Here, we fix $p=0.2, N=100$ and set the learning rate as $\eta_t=0.0001/t$ throughout the iterations. For an equitable comparison, we conduct full-batch gradients for both methods\footnote{It is important to note that the analyses and simulations presented in this paper can be readily extended to scenarios involving stochastic gradient computations by the devices and the central server.}. Furthermore, the local coded datasets in both methods are of the same size.  From Fig.~\ref{fig: scfl}, it is evident that ACFL achieves superior learning performance under the same privacy level, both in the i.i.d. setting (with $\sigma_s^2 = 0$) and in the non-i.i.d. setting (with $\sigma_s^2 = 0.001$), where the non-i.i.d. nature of the data is controlled by the value of $\sigma_s^2$.  Moreover, the learning performance of ACFL exhibits greater resilience and suffers less from increasing the noise variances to heighten the privacy level. This advantage stems from two aspects. First, in ACFL, before the training, the transformed local datasets retain information which is totally equivalent to the raw datasets for gradient computation, contributing to better learning performance at the same privacy level. Second, during the training, SCFL consistently utilizes fixed aggregation weights across iterations at the server, while ACFL uses adaptive aggregation weights to aggregate information at the server by fully considering the generation process of the global coded dataset and the dynamic nature of the trained model over iterations. In this way, ACFL guarantees the optimal performance in terms of both privacy and learning.

\begin{figure}[h]
    \centering
    \subfloat[]{%
        \includegraphics[width=0.8\linewidth]{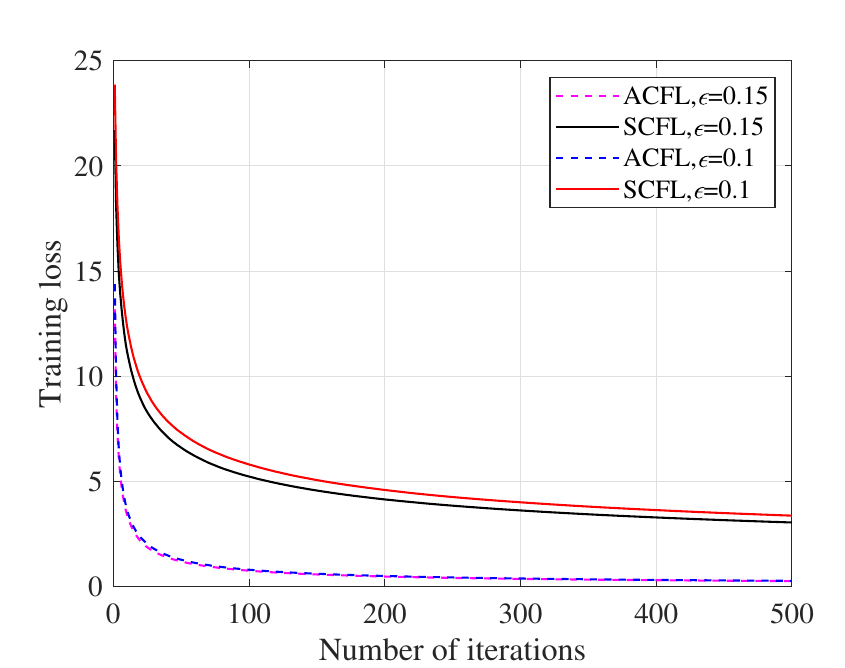}
    }

    \subfloat[]{%
        \includegraphics[width=0.8\linewidth]{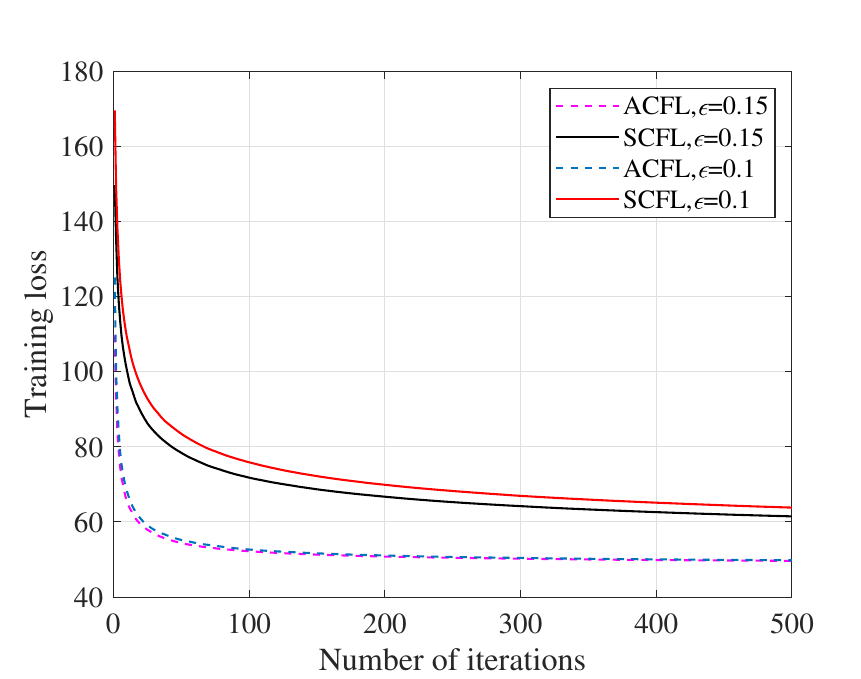}
    }
    \caption{{Training loss as a function of number of iterations of ACFL and SCFL under varying privacy levels and varying degrees of non-i.i.d.-ness in the training data. (a) $\sigma_s^2=0$. (b) $\sigma_s^2=0.001$.} }
    \label{fig: scfl}
\end{figure}
\begin{figure}[h]
    \centering
    \subfloat[]{%
        \includegraphics[width=0.8\linewidth]{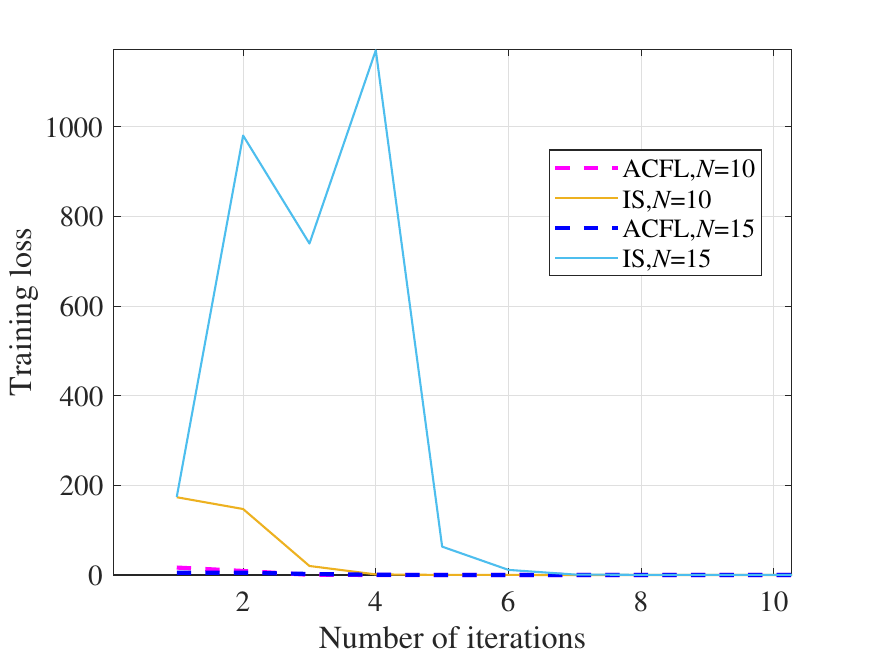}
        \label{fig: compare is our linear}
    }

    \subfloat[]{%
        \includegraphics[width=0.8\linewidth]{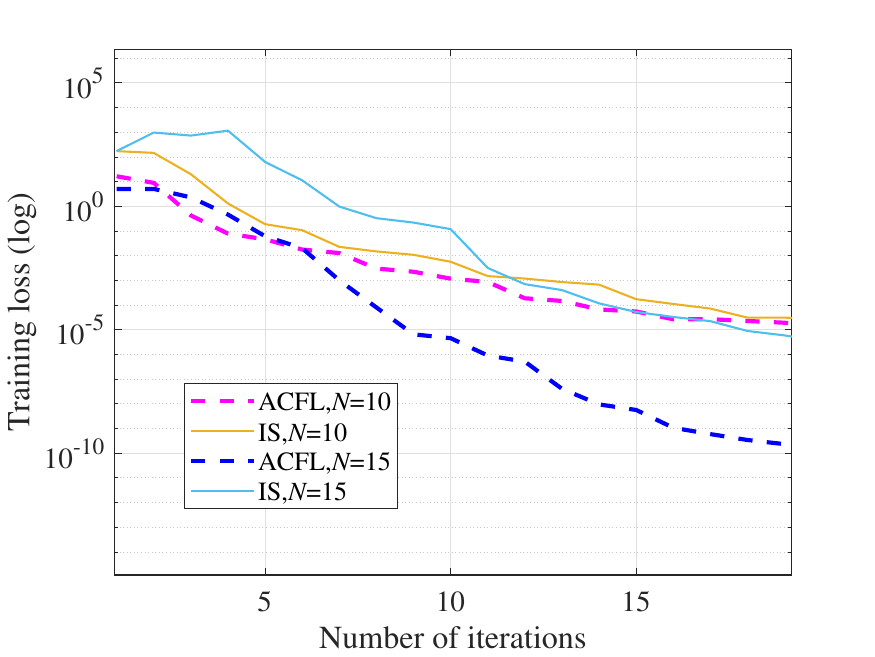}
        \label{fig: compare is our log}
    }
    \caption{Training loss as a function of the number of iterations for ACFL and IS with varying numbers of devices. (a) The training loss is shown on a linear scale. (b) The training loss is shown on a logarithmic scale.
 }
    \label{fig: compare our is}
\end{figure}

Finally, we compare the proposed method with the FL method that simply ignores the existence of stragglers and uses only the information sent by the non-stragglers during the training iterations (denoted as IS). IS is equivalent to fixing the parameter as $\alpha_t=0$ in the proposed method. In Fig.~\ref{fig: compare our is}, we show the training loss as a function of the number of iterations for both the proposed method and IS, with varying numbers of devices, where we fix $\sigma_s^2=0$, $p=0.8$, $\sigma_1=\sigma_2=0.2$, and set the learning rate as $\eta_t=0.01/t$. We can observe in Fig.~\ref{fig: compare our is} that, during the early stages of training, IS experiences very unstable performance due to missing information from the stragglers. In contrast, the proposed method converges stably throughout the training process, benefiting from the fact that the missing information from the stragglers can be compensated by the information gained from the global coded dataset. It is worth noting that the performance gain of ACFL comes at the cost of some privacy leakage compared to IS. In practice, a proper trade-off between learning performance and privacy preservation should be considered when applying ACFL. It is worth emphasizing that the value of our work lies in providing more flexibility to adjust the trade-off between privacy level and straggler mitigation, making the designed FL system more adaptable to different scenarios.

\subsection{{Nonlinear regression scenarios}} 
 In practice, many nonlinear problems can be transformed into linear regression problems, indicating that the proposed method has broader applications beyond linear cases. To verify the value of our proposed method in nonlinear problems, we consider an FL setting with nonlinear training loss, where each device holds a local dataset with $m_i = 100$ samples, $\forall i$. The local dataset at device $i$ is represented by ${{\mathbf{x}}^{(i)}} \in \mathbb{R}^{m_i \times 1}$, where $x^i_k \in \mathbb{R}$ denotes the $k$-th data sample in ${{\mathbf{x}}^{(i)}}$, $k=1,...,m_i$. The elements in ${{\mathbf{x}}^{(i)}}$ are i.i.d., drawn from the uniform distribution over the interval $[-1, 1]$. The label ${{{y}}^i_k}$ corresponding to $x^i_k$ is generated according to $y_k^i = \exp \left( {\xi x_k^i} \right)-1$, $\forall i,k$, where $\exp \left( {\cdot} \right)$ is the exponential function, and $\xi=0.5$ is a constant. To fit the nonlinear relationship between the data samples and the labels, we use a high-order approximation by expanding each training data sample as \(\left[ {x_k^i,{{\left( {x_k^i} \right)}^2},...,{{\left( {x_k^i} \right)}^{10}}} \right]\) and applying the linear regression model expressed as \({{\mathbf{y}}^{\left( i \right)}} = {{\mathbf{X}}^{\left( i \right)}}{\mathbf{w}}\), where ${{\mathbf{y}}^{(i)}} \in \mathbb{R}^{m_i \times 1}$ contains all the labels of the local dataset for device \(i\). 
Here, ${{\mathbf{X}}^{(i)}} \in \mathbb{R}^{m_i \times 10}$ represents the expanded local dataset of device \(i\), where the \(k\)-th row of ${{\mathbf{X}}^{(i)}}$ is given as \(\left[ {x_k^i,{{\left( {x_k^i} \right)}^2},...,{{\left( {x_k^i} \right)}^{10}}} \right]\), $\forall k$.  The vector \(\mathbf{w} \in \mathbb{R}^{10 \times 1}\) is the model parameter to be optimized, which indicates that the devices collaboratively train a linear regression model defined by the vector $\mathbf{w} \in \mathbb{R}^{d \times o}$, where $d = 10$ represents the input feature dimension and $o = 1$ denotes the output dimension. 
Before the training, the initial parameter vector ${\mathbf{w}}_0$ is generated, where the elements are drawn from a uniform distribution in the interval $[0, 1]$. To test the performance of the trained model, a test set is generated consisting of 10 data samples drawn from the same distribution as the training data. The corresponding labels are generated in the same manner as those in the training set.

\begin{figure}[h]
    \centering
    \subfloat[]{%
        \includegraphics[width=0.8\linewidth]{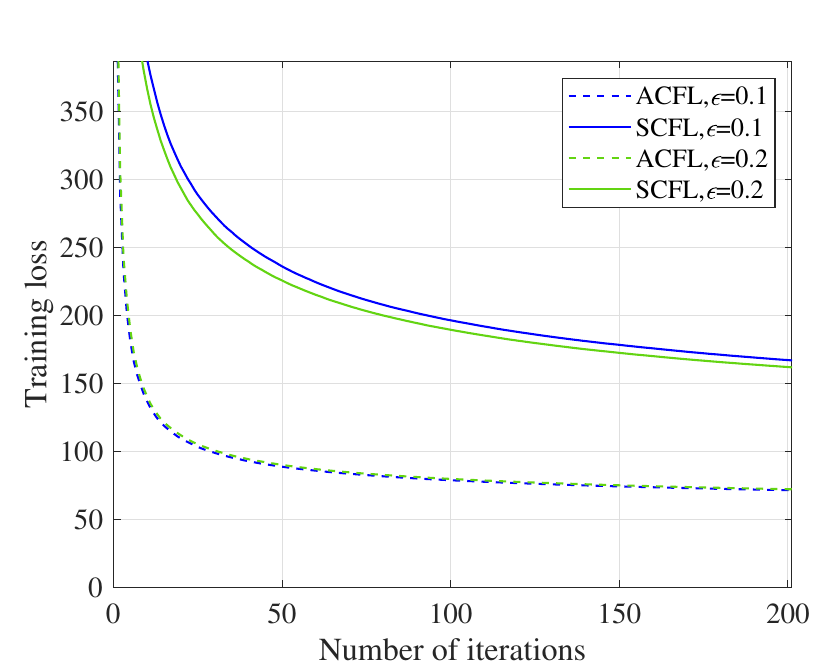}
        \label{fig: non linear train}
    }

    \subfloat[]{%
        \includegraphics[width=0.8\linewidth]{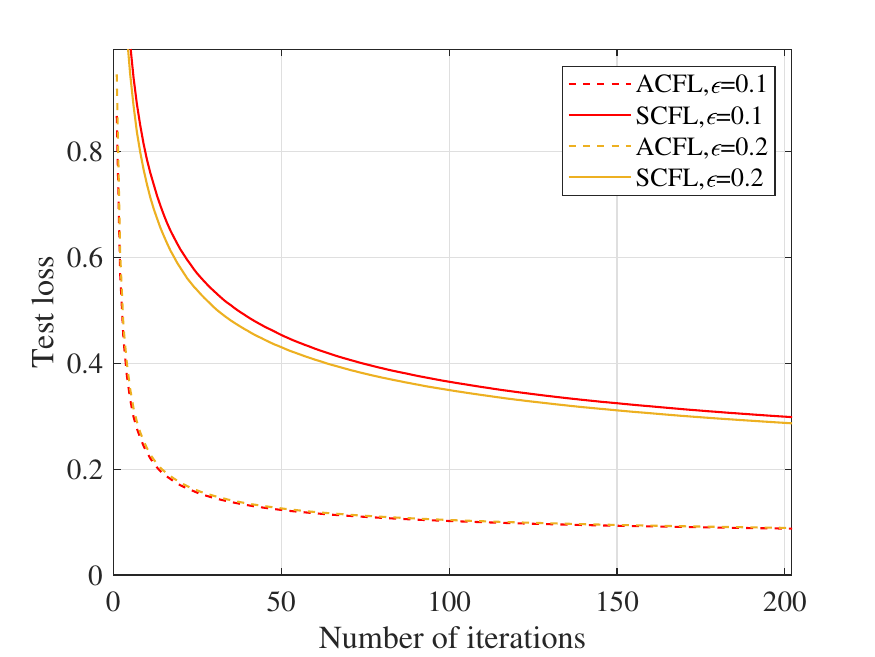}
        \label{fig: non linear test}
    }
    \caption{{Training loss and test loss as functions of the number of iterations for ACFL and SCFL under varying privacy levels, where each device holds a local dataset with $m_i = 100$ samples, $\forall i$. (a) The training loss. (b) The test loss.}
 }
    \label{fig: nonlinear}
\end{figure}

To verify the superiority of the proposed method, we compare its performance with SCFL by plotting the training loss and test loss as functions of the number of iterations under various privacy levels in Fig.~\ref{fig: nonlinear}. Here, we fix $p=0.2$, set $N=100$, and use a learning rate of $\eta_t=0.0001/t$ during training. From Fig.~\ref{fig: non linear train}, we can clearly observe that ACFL achieves significantly better learning performance at the same privacy level, attaining a lower training loss after the same number of iterations and converging faster. Similarly, Fig.~\ref{fig: non linear test} shows that ACFL also outperforms SCFL in terms of test performance. These results demonstrate that the proposed method holds significant value in FL under nonlinear regression scenarios.

   To further demonstrate the superiority of the proposed method when each device possesses more data samples, we compare the performance of ACFL with SCFL by plotting the training and test losses as functions of the number of iterations under various privacy levels in Fig.~\ref{fig: nonlinear1000}. In this setup, each device holds a local dataset with $m_i = 1000$ samples; we set $p=0.2$, $N=100$, and use a learning rate of $\eta_t=0.00001/t$. From Fig.~\ref{fig: nonlinear1000}, it can be observed that ACFL achieves significantly better learning performance under the same privacy level, attaining lower training and test losses after the same number of iterations. These results indicate that ACFL outperforms the baseline across various settings, regardless of whether each device has a smaller or larger amount of training data in FL under nonlinear regression scenarios.

\begin{figure}
    \centering
    \subfloat[]{%
        \includegraphics[width=0.8\linewidth]{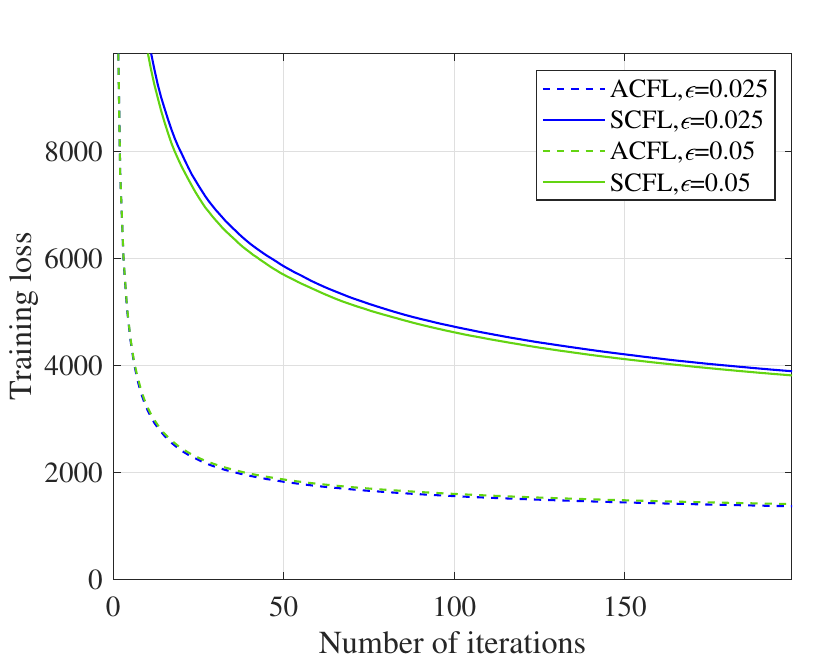}
       
    }

    \subfloat[]{%
        \includegraphics[width=0.8\linewidth]{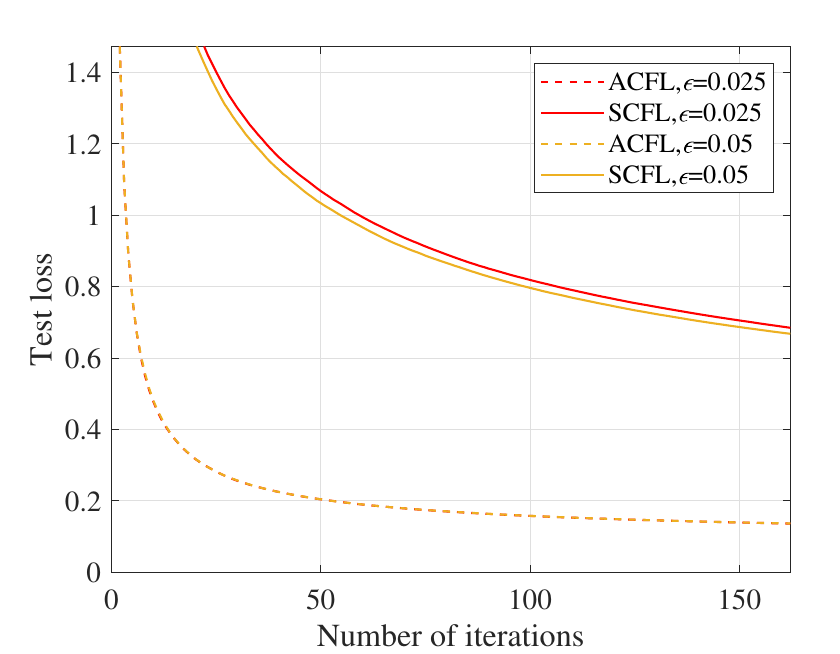}
       
    }
    \caption{{{Training loss and test loss as functions of the number of iterations for ACFL and SCFL under varying privacy levels, where each device holds a local dataset with $m_i = 1000$ samples, $\forall i$. (a) The training loss. (b) The test loss.}}
 }
    \label{fig: nonlinear1000}
\end{figure}

In this simulation, we specifically consider a simple nonlinear regression model based on the power law relationship. This choice is motivated by the fact that power law models have widespread applications, as numerous physical and human-made phenomena approximately follow power law relationships. Examples include modeling in the human motor system \cite{lacquaniti1983law} and the modeling of banking networks \cite{fricke2015distribution}. It is worth emphasizing that the proposed method can also be applied to other nonlinear regression models using various techniques that transform nonlinear regression problems into linear ones. However, due to page limitations, extensive examples and simulations are not provided in this paper. Investigating scalability and practical performance in more complex scenarios will be considered as a promising direction for future research.

\section{Conclusions}
\label{conclusions}
In this article, we addressed the problem of FL in the presence of stragglers. For this issue, the CFL paradigm was explored in the literature. However, the aggregation of gradients at the central server in existing CFL methods, which employ fixed weights across iterations, leads to diminished learning performance. To overcome this limitation, we proposed a new method: ACFL. In ACFL, each device uploads a local coded dataset to the central server before training, enabling the central server to generate a global coded dataset under privacy-preservation requirements. During each iteration of the training, the central server aggregates the received gradients as well as the gradient computed from the global coded dataset by using an adaptive policy to determine the aggregation weights. This approach optimizes the performance in terms of privacy and learning under theoretical guarantees. Finally, we provided simulation results to demonstrate the superiority of ACFL compared to a number of state-of-the-art baseline methods. 

 It is worth noting that in this work, we focused on the privacy preservation of raw datasets of the devices. In the future, we plan to extend the proposed method to more complex scenarios, where private information could also be revealed through the analysis of gradients transmitted from the devices and the devices do not share raw local gradients with the server \cite{wei2020federated, wei2023personalized, lyu2024secure}. Additionally, we plan to propose an enhanced version of the ACFL scheme by allowing devices to perform varying amounts of work in each iteration based on their available resources or using partial computations of the stragglers, rather than directly dropping the stragglers \cite{li2020federated}. We would also like to point out that, in this work, we assumed that the server is trustworthy and operates honestly, as has been similarly assumed in most of the existing FL literature. In the future, we plan to extend this work to more complex environments where the server may be malicious based on verifiable FL techniques \cite{zhang2022towards}.
In future work, we could also explore how multi-modal learning can preserve privacy under the CFL framework by leveraging modality-specific noise injection \cite{du2024distributed}.

\appendices
\section{Proof of Theorem~\ref{Privacy performance of ACFL}}
  \label{appendix privacy}
In the following proof, for the sake of brevity of denotation, we will omit the index $(i)$. Noting that 
\begin{align}
    \label{jk11}
    &\mathop {\sup }\limits_{j,k,\Gamma\left( {{\mathbf{X}},{\mathbf{Y}}} \right)} I\left( {\left. {X_{j,k}^{};{\mathbf{H}}_X^{},{\mathbf{H}}_Y^{}} \right|{{\mathbf{X}}^{ - j,k}},{\mathbf{Y}}} \right)\nonumber\\
    =& \mathop {\sup }\limits_{\Gamma\left( {{\mathbf{X}},{\mathbf{Y}}} \right)} I\left( {\left. {X_{1,1}^{};{\mathbf{H}}_X^{},{\mathbf{H}}_Y^{}} \right|{{\mathbf{X}}^{ - 1,1}},{\mathbf{Y}}} \right),
\end{align}
it is sufficient to show
\begin{align}
    \label{prove_equ}
    \mathop {\sup }\limits_{\Gamma\left( {{\mathbf{X}},{\mathbf{Y}}} \right)} I\left( {\left. {X_{1,1}^{};{\mathbf{H}}_X^{},{\mathbf{H}}_Y^{}} \right|{{\mathbf{X}}^{ - 1,1}},{\mathbf{Y}}} \right) \leq \epsilon,
\end{align}
in order to prove Theorem \ref{Privacy performance of ACFL} without loss of generality. To this end, we can derive 
    \begin{align}
        \label{midp1}
  &\mathop {\sup }\limits_{\Gamma\left( {{\mathbf{X}},{\mathbf{Y}}} \right)} I\left( {\left. {X_{1,1}^{};{\mathbf{H}}_X^{},{\mathbf{H}}_Y^{}} \right|{{\mathbf{X}}^{ - 1,1}},{\mathbf{Y}}} \right) \nonumber \\
   = &\mathop {\sup }\limits_{\Gamma\left( {{\mathbf{X}},{\mathbf{Y}}} \right)} I\left( {\left. {X_{1,1}^{};{{\mathbf{X}}^T}{\mathbf{X}} + {\mathbf{N}}_1^{},{{\mathbf{X}}^T}{\mathbf{Y}} + {\mathbf{N}}_2^{}} \right|{{\mathbf{X}}^{ - 1,1}},{\mathbf{Y}}} \right) \nonumber \\
   = &\mathop {\sup }\limits_{\Gamma\left( {{\mathbf{X}},{\mathbf{Y}}} \right)} h\left( {\left. {{{\mathbf{X}}^T}{\mathbf{X}} + {\mathbf{N}}_1^{},{{\mathbf{X}}^T}{\mathbf{Y}} + {\mathbf{N}}_2^{}} \right|{{\mathbf{X}}^{ - 1,1}},{\mathbf{Y}}} \right) \nonumber \\
   &- h\left( {\left. {{{\mathbf{X}}^T}{\mathbf{X}} + {\mathbf{N}}_1^{},{{\mathbf{X}}^T}{\mathbf{Y}} + {\mathbf{N}}_2^{}} \right|{\mathbf{X}},{\mathbf{Y}}} \right) \nonumber \\
   = &\mathop {\sup }\limits_{\Gamma\left( {{\mathbf{X}},{\mathbf{Y}}} \right)} h\left( {\left. {{{\mathbf{X}}^T}{\mathbf{X}} + {\mathbf{N}}_1^{},{{\mathbf{X}}^T}{\mathbf{Y}} + {\mathbf{N}}_2^{}} \right|{{\mathbf{X}}^{ - 1,1}},{\mathbf{Y}}} \right) \nonumber \\
  & - h\left( {{\mathbf{N}}_1^{},{\mathbf{N}}_2^{}} \right),
  \end{align}
where \(h\left( {\cdot} \right)\) denotes the entropy, the first equality is derived by substituting (\ref{local coded dataset}) and (\ref{coded label}) into (\ref{midp1}) and the second equality holds according to the definition of mutual information. In (\ref{midp1}), for the first term, we have (\ref{first term mi}) based on the independence between the additive noise and the raw data. 
  \begin{figure*} 
\begin{align}
    \label{first term mi}
  &\mathop {\sup }\limits_{\Gamma\left( {{\mathbf{X}},{\mathbf{Y}}} \right)} h\left( {\left. {{{\mathbf{X}}^T}{\mathbf{X}} + {\mathbf{N}}_1^{},{{\mathbf{X}}^T}{\mathbf{Y}} + {\mathbf{N}}_2^{}} \right|{{\mathbf{X}}^{ - 1,1}},{\mathbf{Y}}} \right) \nonumber \\
   =& \mathop {\sup }\limits_{\Gamma\left( {{\mathbf{X}},{\mathbf{Y}}} \right)} h\left( {\left. {\left\{ {{X_{1,j}}{X_{1,1}} + {{\left( {{N_1}} \right)}_{j,1}},j \ne 1} \right\},\left\{ {{X_{1,k}}{X_{1,1}} + {{\left( {{N_1}} \right)}_{1,k}},k \ne 1} \right\},X_{1,1}^2 + {{\left( {{N_1}} \right)}_{1,1}},\left\{ {{X_{1,1}}{Y_{1,k}} + {{\left( {{N_2}} \right)}_{1,k}},\forall k} \right\}} \right|{{\mathbf{X}}^{ - 1,1}},{\mathbf{Y}}} \right) \nonumber \\
   &+ h\left( {\left\{ {{{\left( {{{N}}_1^{}} \right)}_{j,k}},j \ne 1,k \ne 1} \right\}} \right) + h\left( {\left\{ {{{\left( {{{N}}_2^{}} \right)}_{j,k}},j \ne 1,\forall k} \right\}} \right) \nonumber \\
   \leq& \mathop {\sup }\limits_{\Gamma\left( {\mathbf{X}} \right)} \sum\limits_{j \ne 1} {h\left( {\left. {{X_{1,j}}{X_{1,1}} + {{\left( {{N_1}} \right)}_{j,1}}} \right|{{\mathbf{X}}^{ - 1,1}}} \right)}  + \mathop {\sup }\limits_{\Gamma\left( {\mathbf{X}} \right)} \sum\limits_{k \ne 1} {h\left( {\left. {{X_{1,k}}{X_{1,1}} + {{\left( {{N_1}} \right)}_{1,k}}} \right|{{\mathbf{X}}^{ - 1,1}}} \right)} + \mathop {\sup }\limits_{\Gamma\left( {{X_{1,1}}} \right)} h\left( {X_{1,1}^2 + {{\left( {{N_1}} \right)}_{1,1}}} \right)  \nonumber \\
   &+ \mathop {\sup }\limits_{\Gamma\left( {\mathbf{Y}} \right)} \sum\limits_k {h\left( {\left. {{X_{1,1}}{Y_{1,k}} + {{\left( {{N_2}} \right)}_{1,k}}} \right|{\mathbf{Y}}} \right)} + h\left( {\left\{ {{{\left( {{{N}}_1^{}} \right)}_{j,k}},j \ne 1,k \ne 1} \right\}} \right) + h\left( {\left\{ {{{\left( {{{N}}_2^{}} \right)}_{j,k}},j \ne 1,\forall k} \right\}} \right) \nonumber \\
   \leq& \mathop {\sup }\limits_{\Gamma\left( {\mathbf{X}} \right)} \sum\limits_{j \ne 1} {h\left( {{X_{1,j}}{X_{1,1}} + {{\left( {{N_1}} \right)}_{j,1}}} \right)}  + \mathop {\sup }\limits_{\Gamma\left( {\mathbf{X}} \right)} \sum\limits_{k \ne 1} {h\left( {{X_{1,k}}{X_{1,1}} + {{\left( {{N_1}} \right)}_{1,k}}} \right)}+ \mathop {\sup }\limits_{\Gamma\left( {{X_{1,1}}} \right)} h\left( {X_{1,1}^2 + {{\left( {{N_1}} \right)}_{1,1}}} \right)  \nonumber \\
   & + \mathop {\sup }\limits_{\Gamma\left( {\mathbf{Y}} \right)} \sum\limits_k {h\left( {{X_{1,1}}{Y_{1,k}} + {{\left( {{N_2}} \right)}_{1,k}}} \right)} +h\left( {\left\{ {{{\left( {{{N}}_1^{}} \right)}_{j,k}},j \ne 1,k \ne 1} \right\}} \right) + h\left( {\left\{ {{{\left( {{{N}}_2^{}} \right)}_{j,k}},j \ne 1,\forall k} \right\}} \right).
\end{align}
  \end{figure*} 
Substituting (\ref{first term mi}) into (\ref{midp1}) yields 
\begin{align}
    \label{midp2}
 & \mathop {\sup }\limits_{\Gamma\left( {{\mathbf{X}},{\mathbf{Y}}} \right)} I\left( {\left. {X_{1,1}^{};{\mathbf{H}}_X^{},{\mathbf{H}}_Y^{}} \right|{{\mathbf{X}}^{ - 1,1}},{\mathbf{Y}}} \right) \nonumber \\
   \leq& \mathop {\sup }\limits_{\Gamma\left( {\mathbf{X}} \right)} \sum\limits_{j \ne 1} {h\left( {{X_{1,j}}{X_{1,1}} + {{\left( {{N_1}} \right)}_{j,1}}} \right)}  \nonumber \\
  & + \mathop {\sup }\limits_{\Gamma\left( {\mathbf{X}} \right)} \sum\limits_{k \ne 1} {h\left( {{X_{1,k}}{X_{1,1}} + {{\left( {{N_1}} \right)}_{1,k}}} \right)}  \nonumber \\
  & + \mathop {\sup }\limits_{\Gamma\left( {{X_{1,1}}} \right)} h\left( {X_{1,1}^2 + {{\left( {{N_1}} \right)}_{1,1}}} \right) \nonumber \\
  & + \mathop {\sup }\limits_{\Gamma\left( {\mathbf{Y}} \right)} \sum\limits_k {h\left( {{X_{1,1}}{Y_{1,k}} + {{\left( {{N_2}} \right)}_{1,k}}} \right)}  \nonumber \\
  & - h\left( {\left\{ {{{\left( {{{N}}_1^{}} \right)}_{j,k}},j = 1,k \ne 1} \right\}} \right) \nonumber\\
  &- h\left( {\left\{ {{{\left( {{{N}}_1^{}} \right)}_{j,k}},j \ne 1,k = 1} \right\}} \right) \nonumber \\
 &  - h\left[ {{{\left( {{{N}}_1^{}} \right)}_{1,1}}} \right] - h\left( {\left\{ {{{\left( {{{N}}_2^{}} \right)}_{j,k}},j = 1,\forall k} \right\}} \right) \nonumber \\
   \leq& \left( {d - \frac{1}{2}} \right)\log \left[ {2\pi e\left( {1 + \sigma _1^2} \right)} \right] + \frac{o}{2}\log \left[ {2\pi e\left( {1 + \sigma _2^2} \right)} \right] \nonumber \\
  & - \left( {d - \frac{1}{2}} \right)\log \left( {2\pi e\sigma _1^2} \right) - \frac{o}{2}\log \left( {2\pi e\sigma _2^2} \right) \nonumber \\
   = &\left( {d - \frac{1}{2}} \right)\log \frac{{1 + \sigma _1^2}}{{\sigma _1^2}} + \frac{o}{2}\log \frac{{1 + \sigma _2^2}}{{\sigma _2^2}},
\end{align}
based on Assumption \ref{assumption absolute values in x and y} and the maximum entropy bound \cite{cover1999elements}. From (\ref{midp2}), we can easily derive Theorem \ref{Privacy performance of ACFL}. 

\section{Proof of Lemma~\ref{bounded variance}}
\label{appendix bounded variance}
    We can derive that
    \begin{align}
        \label{f2norm}
  &\mathbb{E}\left[ {\left. {\left\| {{\mathbf{G}}_t^{All}} \right\|_F^2} \right|{{\mathbf{W}}_t}} \right]\nonumber\\
  = &\mathbb{E}\left[ {\left. {\left\| {{\alpha _t}{\mathbf{G}}_t^S + \frac{{1 - {\alpha _t}}}{{1 - p}}\sum\limits_{i = 1}^N {{\mathbf{G}}_t^{\left( i \right)}I_t^{\left( i \right)}} } \right\|_F^2} \right|{{\mathbf{W}}_t}} \right] \nonumber \\
   =& \mathbb{E}\left[ {\left\| {\sum\limits_{i = 1}^N {{\mathbf{G}}_t^{\left( i \right)}\left[ {{\alpha _t} + \frac{{1 - {\alpha _t}}}{{1 - p}}I_t^{\left( i \right)}} \right]} } \right.} \right. \nonumber \\
   &+ \left. {\left. {\left. {{\alpha _t}\sum\limits_{i = 1}^N {{\mathbf{N}}_1^{\left( i \right)}} {{\mathbf{W}}_t} - {\alpha _t}\sum\limits_{i = 1}^N {{\mathbf{N}}_2^{\left( i \right)}} } \right\|_F^2} \right|{{\mathbf{W}}_t}} \right] \nonumber \\
   = &\mathbb{E}\left[ {\left. {\left\| {\sum\limits_{i = 1}^N {{\mathbf{G}}_t^{\left( i \right)}\left[ {{\alpha _t} + \frac{{1 - {\alpha _t}}}{{1 - p}}I_t^{\left( i \right)}} \right]} } \right\|_F^2} \right|{{\mathbf{W}}_t}} \right] \nonumber\\
   &+ \mathbb{E}\left[ {\left. {\left\| {{\alpha _t}\sum\limits_{i = 1}^N {{\mathbf{N}}_1^{\left( i \right)}} {{\mathbf{W}}_t}} \right\|_F^2} \right|{{\mathbf{W}}_t}} \right]\nonumber\\
   & + \mathbb{E}\left( {\left\| {{\alpha _t}\sum\limits_{i = 1}^N {{\mathbf{N}}_2^{\left( i \right)}} } \right\|_F^2} \right),
    \end{align}
based on (\ref{local coded dataset}), (\ref{coded label}), (\ref{global coded dataset}), (\ref{server update}), (\ref{global model update}) and the fact that the additive noise is independent from the straggler behaviour of the devices. For the first term in (\ref{f2norm}), we have
\begin{align}
    \label{term1}
  &\mathbb{E}\left[ {\left. {\left\| {\sum\limits_{i = 1}^N {{\mathbf{G}}_t^{\left( i \right)}\left[ {{\alpha _t} + \frac{{1 - {\alpha _t}}}{{1 - p}}I_t^{\left( i \right)}} \right]} } \right\|_F^2} \right|{{\mathbf{W}}_t}} \right] \nonumber \\
   =& \sum\limits_{{i_1} = 1}^N {\sum\limits_{{i_2} = 1}^N {\left\langle {{\mathbf{G}}_t^{\left( {{i_1}} \right)},{\mathbf{G}}_t^{\left( {{i_2}} \right)}} \right\rangle } }  \nonumber \\
  & \times \mathbb{E}\left[ {\left( {{\alpha _t} + \frac{{1 - {\alpha _t}}}{{1 - p}}I_t^{\left( {{i_1}} \right)}} \right)\left( {{\alpha _t} + \frac{{1 - {\alpha _t}}}{{1 - p}}I_t^{\left( {{i_2}} \right)}} \right)} \right]   \nonumber \\
   = &\sum\limits_{{i_1} = 1}^N {\sum\limits_{{i_2} = 1}^N {\left\langle {{\mathbf{G}}_t^{\left( {{i_1}} \right)},{\mathbf{G}}_t^{\left( {{i_2}} \right)}} \right\rangle } }  - \sum\limits_{i = 1}^N {\left\langle {{\mathbf{G}}_t^{\left( i \right)},{\mathbf{G}}_t^{\left( i \right)}} \right\rangle }  \nonumber \\
   &+ \sum\limits_{i = 1}^N {\left\langle {{\mathbf{G}}_t^{\left( i \right)},{\mathbf{G}}_t^{\left( i \right)}} \right\rangle \left[ {\alpha _t^2p + \left( {1 - p} \right){{\left( {{\alpha _t} + \frac{{1 - {\alpha _t}}}{{1 - p}}} \right)}^2}} \right]}  \nonumber \\
   \leq& \left[ {\alpha _t^2p + \left( {1 - p} \right){{\left( {{\alpha _t} + \frac{{1 - {\alpha _t}}}{{1 - p}}} \right)}^2} + N - 1} \right]\sum\limits_{i = 1}^N {\left\| {{\mathbf{G}}_t^{\left( i \right)}} \right\|_F^2}  \nonumber \\
   \leq& \left[ {\alpha _t^2p + \left( {1 - p} \right){{\left( {{\alpha _t} + \frac{{1 - {\alpha _t}}}{{1 - p}}} \right)}^2} + N - 1} \right]N{\beta ^2},
\end{align}
where \(\left\langle {{\mathbf{A}},{\mathbf{B}}} \right\rangle \) denotes the Hadamard Product of two matrices, the first inequality holds due to the following basic inequality:
\begin{align}
    \label{basic1}
    \left\| {\sum\limits_{i = 1}^n {{{\mathbf{A}}_i}} } \right\|_F^2 \leq n\sum\limits_{i = 1}^n {\left\| {{{\mathbf{A}}_i}} \right\|_F^2} ,{\text{ }}\forall {{\mathbf{A}}_i} \in {\mathbb{R}^{a \times b}},
\end{align}
and the second inequality holds due to Assumption \ref{assumption bounded gradients}. For the second term in (\ref{f2norm}), we have
\begin{align}
    \label{term2}
    \mathbb{E}\left[ {\left. {\left\| {{\alpha _t}\sum\limits_{i = 1}^N {{\mathbf{N}}_1^{\left( i \right)}} {{\mathbf{W}}_t}} \right\|_F^2} \right|{{\mathbf{W}}_t}} \right] &= \alpha _t^2\sum\limits_{i = 1}^N {\mathbb{E}\left[ {\left. {\left\| {{\mathbf{N}}_1^{\left( i \right)}{{\mathbf{W}}_t}} \right\|_F^2} \right|{{\mathbf{W}}_t}} \right]}  \nonumber \\
   &= \alpha _t^2Nd\sigma _1^2\left\| {{{\mathbf{W}}_t}} \right\|_F^2 \nonumber \\
   &\leq \alpha _t^2Nd\sigma _1^2{C^2},
\end{align}
based on the independence of the additive noise on different local datasets and Assumption \ref{assumption bounded parameter}. For the third term in (\ref{f2norm}), we have
\begin{align}
    \label{term3}
    \mathbb{E}\left( {\left\| {{\alpha _t}\sum\limits_{i = 1}^N {{\mathbf{N}}_2^{\left( i \right)}} } \right\|_F^2} \right) = \alpha _t^2N\sigma _2^2od.
\end{align}
Substituting (\ref{term1}), (\ref{term2}) and (\ref{term3}) into (\ref{f2norm}) yields Lemma \ref{bounded variance}. 

\bibliographystyle{IEEEtran}   
\bibliography{reference}      
\end{document}